\documentclass[a4paper,UKenglish,cleveref, autoref, thm-restate]{lipics-v2021}
\pdfoutput=1 

\title{Computing Inductive Invariants of Regular Abstraction Frameworks}

\author{Philipp {Czerner}}{Technical University of Munich, Germany}{}{https://orcid.org/0000-0002-1786-9592}{}
\author{Javier Esparza}{Technical University of Munich, Germany}{}{https://orcid.org/0000-0001-9862-4919}{}
\author{Valentin Krasotin}{Technical University of Munich, Germany}{}{https://orcid.org/0009-0002-2129-2754}{}
\author{Christoph Welzel-Mohr}{Technical University of Munich, Germany}{}{https://orcid.org/0000-0001-5583-0640}{}

\authorrunning{P.\ Czerner, J.\ Esparza, V.\ Krasotin and C.\ Welzel-Mohr} 

\Copyright{Philipp Czerner, Javier Esparza, Valentin Krasotin and Christoph Welzel-Mohr} 



\category{} 

\relatedversion{} 

\nolinenumbers 

\EventEditors{Rupak Majumdar and Alexandra Silva}
\EventNoEds{2}
\EventLongTitle{35th International Conference on Concurrency Theory (CONCUR 2024)}
\EventShortTitle{CONCUR 2024}
\EventAcronym{CONCUR}
\EventYear{2024}
\EventDate{September 9--13, 2024}
\EventLocation{Calgary, Canada}
\EventLogo{}
\SeriesVolume{311}
\ArticleNo{20}

\newcommand{\appref}[1]{Appendix~\ref{#1}}

\usepackage{amsmath}
\usepackage{complexity}
\usepackage{tikz}
\usepackage[vlined,linesnumbered]{algorithm2e}
\usepackage{csquotes}
\usepackage{nicefrac}
\usepackage{array}
\usepackage{ragged2e}
\usepackage{bm}
\newcolumntype{R}[1]{>{\RaggedLeft\arraybackslash}p{#1}}
\usetikzlibrary{arrows.meta,automata,positioning,petri}

\tikzset{>={Latex[length=2mm,width=2mm]}}

\renewcommand{\O}{O}

\newcommand{\system}{\mathcal{S}}

\newcommand{\interpretation}{\mathcal{V}}
\newcommand{\configurations}{\mathcal{C}}
\newcommand{\constraints}{\mathcal{A}}
\newcommand{\constraint}{A}

\newcommand{\trafun}{\Delta}
\newcommand{\reach}{\textit{Reach}}
\newcommand{\potreach}{\textit{PReach}}

\newcommand{\comp}[1]{\overline{{#1}}}

\newcommand{\N}{\mathbb{N}}

\newcommand{\vari}[3]{{#1}_{{#2}:{#3}}}
\newcommand{\Inductive}{\textit{Ind}}
\newcommand{\btape}{\#}

\newcommand{\size}[1]{\left| #1 \right|}
\newcommand{\bk}{B}
\newcommand{\Language}[1]{\lang{L}({#1})}
\newcommand{\Relation}[1]{\lang{R}({#1})}
\newcommand{\framework}{\mathcal{F}}

\newcommand{\set}[1]{\left\{ #1 \right\}}

\newcommand{\vtuple}[2]{\big[{{#1} \atop {#2}}\big]}

\newcommand{\TM}{\mathcal{M}}
\newcommand{\Confalph}{\Lambda}
\newcommand{\Tapealph}{\Gamma'}
\newcommand{\Run}{\bm{\alpha}}
\newcommand{\productofprimes}{m}
\newcommand{\Markhead}{w}
\newcommand{\Markers}{\mathbf{m}}

\definecolor{nicebg}{HTML}{f6f0e4}
\definecolor{nicered}{HTML}{7f0a13}
\definecolor{nicebgred}{HTML}{f2e7e8}
\definecolor{niceblue}{HTML}{104354}
\definecolor{nicebgblue}{HTML}{e8edee}
\definecolor{nicegreen}{HTML}{217516}
\definecolor{nicebggreen}{HTML}{e9f1e8}
\definecolor{nicepurple}{HTML}{884bab}
\definecolor{nicebgpurple}{HTML}{f3edf7}
\definecolor{niceorange}{HTML}{d27c11}
\definecolor{nicebgorange}{HTML}{fbf2e8}
\definecolor{nicepink}{HTML}{e95f9f}
\definecolor{nicebgpink}{HTML}{fdeff6}
\definecolor{niceredlight}{HTML}{c9888d}
\definecolor{nicebluelight}{HTML}{78a4b8}
\definecolor{nicegreenlight}{HTML}{76de68}
\definecolor{nicepurplelight}{HTML}{bc87db}
\definecolor{niceredbright}{HTML}{bd0310}
\definecolor{nicebgredbright}{HTML}{f9e6e8}
\definecolor{nicebluebright}{HTML}{197b9b}
\definecolor{nicebgbluebright}{HTML}{e8f2f5}

\makeatletter
\newcommand\xleftrightarrow[2][]{\ext@arrow 0099{\longleftrightarrowfill@}{#1}{#2}}
\def\longleftrightarrowfill@{\arrowfill@\leftarrow\relbar\rightarrow}
\makeatother

\begin{document}

\maketitle

\begin{abstract}
Regular transition systems (RTS) are a popular formalism for modeling infinite-state systems in general, and parameterised systems in particular.
In a CONCUR 22 paper, Esparza et al.\ introduce a novel approach to the verification of RTS, based on inductive invariants. The approach computes the intersection of all inductive invariants of a given RTS that can be expressed as CNF formulas with a bounded number of clauses, and uses it to construct an automaton recognising an overapproximation of the reachable configurations. The paper shows that the problem of deciding if the language of this automaton intersects a given regular set of unsafe configurations is in \EXPSPACE\ and \PSPACE-hard.

We introduce \emph{regular abstraction frameworks}, a generalisation of the approach of Esparza et al., very similar to the regular abstractions of Hong and Lin. A framework consists of a regular language of \emph{constraints}, and a transducer, called the  \emph{interpretation}, that assigns to each constraint the set of configurations of the RTS satisfying it. Examples of regular abstraction frameworks include the formulas of Esparza et al., octagons, bounded difference matrices, and views. We show that the generalisation of the decision problem above to regular abstraction frameworks remains in \EXPSPACE, and prove a matching (non-trivial) \EXPSPACE-hardness bound.

\EXPSPACE-hardness implies that, in the worst case, the automaton recognising the overapproximation of the reachable configurations has a double-exponential number of states. We introduce a learning algorithm that computes this automaton in a lazy manner, stopping whenever the current hypothesis is already strong enough to prove safety. We report on an implementation and show that our experimental results improve on those of Esparza et al.
\end{abstract}

\section{Introduction}
Regular transition systems (RTS) are a popular formalism for modelling infinite-state systems satisfying the following conditions: configurations can be encoded as words,  the set of initial configurations is recognised by a finite automaton, and the transition relation is recognised by a transducer. Model checking RTS has been intensely studied under the name of \emph{regular model checking} (see~\cite{JonssonN00,BouajjaniJNT00,KestenMMPS01,BoigelotLW03} and the surveys~\cite{AbdullaJNS04,Abdulla12,AbdullaST18,Abdulla21}). 
Most regular model checking algorithms address the \emph{safety problem}:  given a regular set of unsafe configurations, decide if its intersection with the set of reachable configurations is empty or not. They combine algorithms for the computation of increasingly larger regular subsets of the reachable configurations with acceleration, abstraction, and widening techniques~\cite{BouajjaniJNT00,JonssonN00,DamsLS01,AbdullaJNd02,BoigelotLW03,BouajjaniHV04,BouajjaniT12,BouajjaniHRV12,Legay12,ChenHLR17}. 

Recently,  Esparza et al.\ have introduced  a novel approach that, starting with the set of all configurations of the RTS, computes increasingly smaller inductive invariants, that is, inductive supersets of the reachable configurations.  More precisely, \cite{ERW22} considers invariants given by Boolean formulas in conjunctive normal form with at most $b$ clauses. The paper proves that, for every bound $b \geq 0$, the intersection of \emph{all} inductive $b$-invariants of the system is recognised by a DFA of double exponential size in the RTS. As a corollary, they obtain that, for every $b \geq 0$, deciding if this intersection contains some unsafe configuration is in \EXPSPACE. They also show that the problem is \PSPACE-hard, and leave the question of closing the gap open.

In~\cite{ERW23} (a revised version of~\cite{ERW22}), the \EXPSPACE\ proof is conducted in a more general setting than in~\cite{ERW22}. Inspired by this, 
in our first contribution we show that the approach of~\cite{ERW22} can be vastly generalised to arbitrary \emph{regular abstraction frameworks}, consisting of a  regular language of \emph{constraints}, and an \emph{interpretation}. Interpretations are functions, represented by transducers,  that assign to each constraint a set of configurations, viewed as the set of configurations that \emph{satisfy} the constraint. Examples of regular abstraction frameworks include the formulas of~\cite{ERW22} for every $b \geq 0$, views~\cite{AbdullaHH16}, and families of Presburger arithmetic formulas like octagons~\cite{Mine06} or bounded difference matrices~\cite{KroeningS16,BarrettT18}. A framework induces an abstract interpretation, in which, loosely speaking,  the word encoding a constraint is the abstraction of the set of configurations satisfying the constraint. 
Just as regular model checking started with the observation that different classes of \emph{systems} could be uniformly modeled as RTSs~\cite{AbdullaJNS04,Abdulla12,AbdullaST18,Abdulla21}, we add the observation, also made in~\cite{HongL24},  that different classes of \emph{abstractions} can be uniformly modeled as regular abstraction frameworks. We show that the generalisation of the verification problem of~\cite{ERW22,ERW23} to arbitrary regular abstraction frameworks remains in \EXPSPACE.

In our second contribution we show that our problem is also \EXPSPACE-hard. The reduction (from the acceptance problem for exponentially bounded Turing machines) is surprisingly involved. Loosely speaking, it requires to characterise the set of prefixes of the run of a Turing machine on a given word as an intersection of inductive invariants of a very restrictive kind. We think that this construction can be of independent interest.

Our third and final contribution is motivated by the \EXPSPACE-hardness result. A consequence of this lower bound is that the automaton recognising the overapproximation of the reachable configurations must necessarily have a double-exponential number of states in the worst case. We present an approach, based on automata learning, that constructs increasingly larger automata that recognise increasingly smaller overapproximations, and checks whether they are precise enough to prove safety. A key to the approach is solving the separability problem: given a pair $(c, c')$ of configurations, is there an inductive constraint that \emph{separates} $c$ and $c'$, i.e.\ is satisfied by $c$ but not by $c'$? We show that the problem is \PSPACE-complete and \NP-complete for interpretations captured by length-preserving transducers. We provide an implementation on top of a SAT solver for the latter case (this is the only case considered in \cite{ERW22,ERW23}). An experimental comparison shows that this approach beats the one of~\cite{ERW22,ERW23}.

\subparagraph*{Related work.} As mentioned above, our first contribution is a reformulation of results of~\cite{ERW23} into a more ambitious formalism; it is a conceptual but not a technical novelty. The second and third contributions are new technical results. 

Our regular abstraction frameworks are in the same spirit as the regular abstractions of  Hong and Lin~\cite{HongL24}, which  use regular languages as abstract objects. In this paper we concentrate on the inductive invariant approach of~\cite{ERW22}, and in particular on its complexity. This is unlike the approach of~\cite{HongL24}, which on the one hand is more general, since it also considers liveness properties, but on the other hand does not contain complexity results.

Automata learning has been explored for the verification of regular transition systems multiple times~\cite{Neider14,AbhayPhD,ChenHLR17,VardhanSVA04,NeiderJ13}.
Roughly speaking, all these approaches formulate a learning process to obtain a \emph{regular} inductive invariant of the system that proves a safety property. Since it is impossible to algorithmically identify the cases where such regular inductive invariant exists, timeouts~\cite{ChenHLR17} and resource limits~\cite{Neider14} are used as heuristics. In contrast, our approach is designed to always terminate. In particular, we either provide a regular set of constraints that suffices to establish the safety property or a pair of configurations that cannot be separated by inductive constraints of the considered framework. This information can be used to design a more precise framework by adding a new type of constraints.

\section{Preliminaries and regular transition systems}
\label{sec:prelims}

\subparagraph*{Automata.} Let $\Sigma$ be an alphabet. A \emph{nondeterministic finite automaton (NFA)} over $\Sigma$ is a tuple $A = (Q,\Sigma,\delta,Q_0,F)$ where $Q$ is a finite set of \emph{states}, $\delta:Q \times \Sigma \to \mathcal{P}(Q)$ is the \emph{transition function}, $Q_0 \subseteq Q$ is the set of \emph{initial states}, and $F \subseteq Q$ is the set of \emph{final states}. A \emph{run} of $A$ on a word $w = w_1 \cdots w_l \in \Sigma^l$ is a sequence $q_0q_1 \cdots q_l$ of states where $q_0 \in Q_0$ and $\forall i \in [l]: q_i \in \delta(q_{i-1},w_i)$. A run on $w$ is \emph{accepting} if $q_l \in F$, and $A$ \emph{accepts} $w$ if there exists an accepting run of $A$ on $w$. The language \emph{recognised} by $A$, denoted $L(A)$ or $L_A$, is the set of words accepted by $A$. If $|Q_0| = 1$ and $|\delta(q,a)| = 1$ for every $q\in Q, a \in \Sigma$ $|Q_0| = 1$, then $A$ is a deterministic finite automaton \emph{(DFA)}. In this case, we write $\delta(q,a) = q'$ instead of $\delta(q,a) = \{q'\}$ and have a single initial state $q_0$ instead of a set $Q_0$.

\subparagraph*{Relations.}  Let $R \subseteq X \times Y$ be a relation. The \emph{complement} of $R$ is the relation $\overline{R} := \{(x,y) \in X \times Y \mid (u,w) \notin R\}$. The \emph{inverse} of $R$ is the relation $R^{-1} := \{(y,x) \in Y \times X \mid (x,y) \in R\}$. The \emph{projections} of $R $ onto its first and second components are the sets $R|_1 := \{x \in X\mid \exists y \in Y \colon (x,y) \in R\}$ and $R|_2 := \{y \in Y\mid \exists x \in X \colon (x,y) \in R\}$. The \emph{join} of two relations $R \subseteq X \times Y$ and $S \subseteq Y \times Z$ is the relation $R \circ S := \{(x,z) \in X \times Z \mid \exists y \in Y \colon (x,y) \in R, (y,z) \in S\}$. The \emph{post-image} of a set $X' \subseteq X$ under a relation $R \subseteq X \times Y$, denoted $X' \circ R$ or $R( X')$, is the set $\{y \in Y \mid \exists x \in X' \colon (x,y) \in R\}$; the \emph{pre-image}, denoted $R \circ Y$ or $R^{-1}(Y)$, is defined analogously. Throughout this paper, we only consider relations where $X =\Sigma^*$ and $Y = \Gamma^*$ for some alphabets $\Sigma$, $\Gamma$. We just call them relations. A relation $R \subseteq \Sigma^* \times \Gamma^*$ is \emph{length-preserving} if $(u, w) \in R$ implies $|u|=|w|$.

\subparagraph*{Convolutions and transducers.}  Let $\Sigma$, $\Gamma$ be alphabets, let $\# \notin \Sigma \cup \Gamma$ be a padding symbol, and let $\Sigma_\#:= \Sigma \cup \{\#\}$ and $\Gamma_\#:= \Gamma \cup \{\#\}$. The \emph{convolution} of two words $u = a_1 \ldots a_k \in \Sigma^*$ and $w = b_1 \ldots b_l\in\Gamma^*$, denoted $\vtuple{u}{w}$,  is the word over the alphabet $\Sigma_\#  \times \Gamma_\#$ defined as follows. Intuitively, $\big[{u \atop w}\big]$ is the result of putting $u$ on top of $w$, aligned left, and padding the shorter of $u$ and $w$ with $\#$. Formally,  if $k \leq l$, then $\big[{u \atop w}\big] = \big[{a_1 \atop b_1}\big] \cdots \big[{a_k \atop b_k}\big]\big[{\# \atop b_{k+1}}\big] \cdots
\big[{\# \atop b_{l}}\big]$, and otherwise $\big[{u \atop w}\big] = \big[{a_1 \atop b_1}\big] \cdots \big[{a_l \atop b_l}\big]\big[{a_{l+1} \atop \#}\big] \cdots
\big[{a_k \atop \#}\big]$. The convolution of a tuple of words $u_1 \in \Sigma_1^*, \ldots, u_k \in \Sigma_k^*$ is defined analogously, putting all $k$ words on top of each other, aligned left, and padding the shorter words with $\#$.

A \emph{transducer} over $\Sigma\times\Gamma$ is an NFA over $\Sigma_\#  \times \Gamma_\#$. The binary relation recognised by a transducer $T$ over $\Sigma\times\Gamma$, denoted $\Relation{T}$, is the set of pairs $(u, w) \in \Sigma^* \times \Gamma^*$ such that $T$ accepts $\big[{u \atop w}\big]$. The definition is generalised to relations of higher arity in the obvious way. In the paper transducers recognise binary relations unless mentioned otherwise. A relation is \emph{regular} if it is recognised by some transducer. A transducer  is \emph{length-preserving} if it recognises a length-preserving relation.

\subparagraph*{Complexity of operations on automata and transducers.} Given NFAs $A_1$, $A_2$ over $\Sigma$ with $n_1$ and $n_2$ states,  DFAs $B_1$, $B_2$ over $\Sigma$ with $m_1$ and $m_2$ states, and transducers $T_1$ over $\Sigma \times \Gamma$ and  $T_2$ over $\Gamma \times \Sigma$ with $l_1$ and $l_2$ states, the following facts are well known (see e.g.\ chapters 3 and 5 of~\cite{EB23}):
\begin{itemize}
\item there exist NFAs for $\Language{A_1} \cup \Language{A_2}$,  $\Language{A_1} \cap \Language{A_2}$, and $\overline{\Language{A_1}}$ with at most $n_1 + n_2$, $n_1 n_2$, and $2^{n_1}$ states, respectively; 
\item there exist DFAs for $\Language{B_1} \cup \Language{B_2}$,  $\Language{B_1} \cap \Language{B_2}$, and $\overline{\Language{B_1}}$ with at most $m_1 m_2$, $m_1 m_2$, and $m_1$ states, respectively;
\item there exist NFAs for $\Relation{T_1}|_1$ and $\Relation{T_1}|_2$ and a transducer for $\Relation{T_1}^{-1}$ with at most $l_1$ states;
\item there exists a transducer for $\Relation{T_1} \circ \Relation{T_2}$ with at most $l_1 l_2$ states; and
\item there exist NFAs for $\Language{A_1} \circ \Relation{T_1}$ and $\Relation{T_1} \circ \Language{A_2}$ with at most $n_1 l_1$ and $l_1 n_2$ states, respectively.
\end{itemize}

\subsection{Regular transition systems}
We recall standard notions about regular transition systems and fix some notations. A \emph{transition system} is a pair $\system=(\configurations,\trafun)$ where $\configurations$ is the set of all possible \emph{configurations} of the system, and $\trafun \subseteq \configurations \times \configurations$ is a \emph{transition relation}.  The \emph{reachability relation} $\reach$ is the reflexive and transitive closure of $\trafun$. Observe that, by our definition of post-set, $\trafun(C)$ and $\reach(C)$ are the sets of configurations reachable in one step and in arbitrarily many steps from $C$, respectively.

Regular transition systems are transition systems where $\trafun$ can be finitely represented by a transducer. Formally:

\begin{definition}
 A transition system $\system = (\configurations, \trafun)$ is \emph{regular} if $\configurations$ is a regular language over some alphabet $\Sigma$, and $\trafun$ is a regular relation. We abbreviate regular transition system to RTS.
\end{definition}

RTSs are often used to model parameterised systems~\cite{AbdullaJNS04,Abdulla12,AbdullaST18,Abdulla21}.  In this case, $\Sigma$ is the set of possible \emph{states} of a process, the set of configurations is $\configurations = \Sigma^* \setminus \{\varepsilon\}$, and a configuration $a_1 \cdots a_n \in \Sigma^*$ describes the global state of an \emph{array} consisting of $n$ identical copies of the process, with the $i$-th process in state $a_i$ for every $1 \leq i \leq n$. The transition relation $\trafun$ describes the possible transitions of all arrays, of any length.  

\begin{example}[Token passing~\cite{AbdullaJNS04}]\label{ex:tokenpassing}
We use a version of the  well-known token passing algorithm as running example. We have an array of processes of arbitrary length. At each moment in time, a process either has a token ($t$) or not $(n)$. Initially, only the first process has a token. A process that has a token can pass it to the process to the right if that process does not have one. We set $\Sigma = \{t, n\}$, and so $\configurations = \{t,n\}^* \setminus \{\varepsilon\}$. We have $c_2 \in \trafun(c_1)$ if{}f the word $\vtuple{c_1}{c_2}$ belongs to the regular expression $\big(\vtuple{n}{n}+\vtuple{t}{t}\big)^{*} \; \big(\vtuple{t}{n}\vtuple{n}{t}\big)\; \big(\vtuple{n}{n}+\vtuple{t}{t}\big)^{*}$. For the set of initial configurations $C_I:=tn^*$ where only the first process has a token, the set of reachable configurations is $\reach(C_I)=n^{*} \; t \; n^{*}$.
\end{example}

\section{Regular abstraction frameworks}

In the same way that RTSs can model multiple classes of \emph{systems} (e.g.\ parameterised systems with synchronous/asynchronous, binary/multiway/broadcast communication), regular abstraction frameworks are a formalism to model a wide range of \emph{abstractions}. 

\begin{definition}
An \emph{abstraction framework} is a triple $\framework = (\configurations,\constraints,\interpretation)$, where $\configurations$ is a set of \emph{configurations}, $\constraints$ is a set of \emph{constraints}, and $\interpretation \subseteq \constraints \times \configurations$ is an \emph{interpretation}. $\framework$ is \emph{regular} if $\configurations$ and $\constraints$ are regular languages over alphabets $\Sigma$ and $\Gamma$, respectively, and the interpretation $\interpretation$ is a regular relation over $\constraints \times \configurations$. 
\end{definition}

\noindent Intuitively, the constraints of an abstraction framework are the abstract objects of the abstraction, and $\interpretation(\constraint)$ is the set of configurations abstracted by $\constraint$. The following remark formalises this.

\begin{remark}
An abstraction framework $\framework=(\configurations,\constraints,\interpretation)$ induces an abstract interpretation as follows. The concrete and abstract domains are  $(2^\configurations, \leq_\configurations)$ and $(2^\constraints, \leq_\constraints)$, respectively, where $\leq_\configurations := \subseteq$ and $\leq_\constraints := \supseteq$. Both are complete lattices. The concretisation function $\gamma \colon 2^\constraints \to 2^\configurations$ and the abstraction function $\alpha \colon 2^\configurations \to 2^\constraints$ are given by:
\begin{itemize}
\item $\gamma(\constraints') := \bigcap_{A \in \constraints'} \interpretation(A)$. Intuitively, $\gamma(\constraints')$ is the set of configurations that satisfy all constraints of $\constraints'$. In particular, $\gamma(\emptyset) = \configurations$.
\item $\alpha(\configurations') := \{ A \in \constraints \mid \configurations' \subseteq \interpretation(A) \}$. Intuitively, $\alpha(\configurations')$ is the set of constraints satisfied by all configurations in $\configurations'$. In particular, $\alpha(\emptyset) = \constraints$.
\end{itemize}
It is easy to see that the functions $\alpha$ and $\gamma$ form a Galois connection, that is, for all $C \subseteq \configurations$ and $A \subseteq \constraints$, we have $\mathcal{B} \subseteq \alpha(C) \Leftrightarrow C \subseteq \gamma(\mathcal{B})$. 
\end{remark}

Regular abstractions can be combined to yield more precise ones. Given abstraction frameworks $\framework_1 = (\configurations, \constraints_1, \interpretation_1)$ and $\framework_2=(\configurations, \constraints_2, \interpretation_2)$, we can define new frameworks $(\configurations, \constraints, \interpretation)$ by means of the following operations:
\begin{itemize}
\item Union: $\constraints := \constraints_1 \cup \constraints_2$, $\interpretation(\constraint) := \interpretation_1(\constraint) \mbox{ if $\constraint \in \constraints_1$, else } \interpretation_2(\constraint)$. \\
A constraint of the union framework is either a constraint of the first framework, or a constraint of the second.
\item Convolution: $\constraints = \constraints_1 \times \constraints_2$, $\interpretation(\constraint_1, \constraint_2):= 
\interpretation_1(\constraint_1) \cap \interpretation_2(\constraint_2)$. \\
A constraint of the convolution framework is the conjunction of two constraints, one of each framework. This operation is implicitly used in~\cite{ERW22}: the constraint for a Boolean formula  with $b$ clauses is the convolution, applied $b$ times,  of the constraints for formulas with one clause.
\end{itemize}
The proof of the following lemma is in \appref{app:operations}.
\begin{restatable}{lemma}{operations}\label{lem:operations}
Regular abstraction frameworks are closed under union and convolution. If the interpretations of the frameworks are recognised by transducers with $n_1$ and $n_2$ states, then the interpretations of the union and convolution frameworks are recognised by transducers with $O(n_1+n_2)$ and $O(n_1n_2)$ states, respectively.
\end{restatable}

Many abstractions used in the literature can be modeled as regular abstraction frameworks. We give some examples.

\begin{example}
Consider a transition system where $\configurations = \N^d$ for some $d$, and $\trafun$ is given by a formula of Presburger arithmetic $\delta(\mathbf{x}, \mathbf{x'})$, that is, $(\mathbf{n}, \mathbf{n'}) \in \trafun$ if{}f  $\delta(\mathbf{n}, \mathbf{n'})$ holds. It is well-known that for any Presburger formula there is a transducer recognising the set of its solutions when numbers are encoded in binary, and so with this encoding $(\configurations, \trafun)$ is an RTS.  Any Presburger formula $\varphi(\mathbf{x}, \mathbf{y})$, where $\mathbf{x}$ has dimension $d$ and $\mathbf{y}$ has some arbitrary dimension $e$, induces a regular abstraction framework as follows. The set of constraints is the set of all tuples $\mathbf{m} \in \N^e$; the interpretation assigns to $\mathbf{m}$ all tuples $\mathbf{n}$ such that $\varphi(\mathbf{n}, \mathbf{m})$ holds. Intuitively, the constraints are the formulas $\varphi_\mathbf{m}(\mathbf{x}):= \varphi(\mathbf{x},\mathbf{m})$, but using $\mathbf{m}$ as encoding of $\varphi_\mathbf{m}$.

Special cases of this setting are used in many different areas. For example, bounded difference matrices (see e.g.~\cite{KroeningS16,BarrettT18}) and octagons~\cite{Mine06}  correspond to abstraction frameworks with constraints $\varphi(x_1,x_2, y)$ of the form $x_1 \pm x_2 \leq y$.
\end{example}

\begin{example}
\label{ex:traps}
The approach to regular model checking of~\cite{ERW22} is another instance of a regular abstraction framework. The paper encodes sets of configurations as positive Boolean formulas in conjunctive normal form with a bounded number $b$ of clauses.  We explain this by means of an example. Consider an RTS with $\Sigma = \{a,b,c\}$ and $\configurations = \Sigma^*$. Consider the formula $\varphi=(\vari{a}{1}{5} \vee \vari{b}{1}{5} \vee \vari{a}{3}{5}) \wedge \vari{b}{4}{5}$. We interpret $\varphi$ on configurations. The intended meaning of a literal, say $\vari{a}{1}{5}$, is ``if the configuration has length 5, then its first letter is an $a$.'' So the set of configurations satisfying the formula is $\Sigma^{\leq 4} + \Sigma^6\Sigma^* + (a+b)\Sigma^2 b \Sigma + \Sigma^2 a b\Sigma$. In the formulas of  \cite{ERW22} all literals have the same length, where the length of a literal $\vari{x}{i}{j}$ is $j$.

Formulas with at most $b$ clauses can be encoded as words over the alphabet $\Gamma = (2^\Sigma)^b$. Each clause is encoded as a word over $2^\Sigma$. For example, the encodings of the clauses $(\vari{a}{1}{5} \vee \vari{b}{1}{5} \vee \vari{a}{3}{5})$ and $\vari{b}{4}{5}$ are $\{a,b\} \emptyset \{a\} \emptyset \emptyset$ and $\emptyset\emptyset\emptyset \{b\}\emptyset$, and the encoding of $\varphi$ is the convolution of the encodings of the clauses. It is easy to see that the interpretation of~\cite{ERW22} that assigns to a formula the set of configurations satisfying it is a regular relation recognised by a transducer with $2^b$ states \cite{ERW22}. In particular, for the case $b=1$ we get the two-state transducer on the left of Figure~\ref{fig:interp}.
\end{example}

\begin{example}
\label{ex:third}
In~\cite{AbdullaHH16} Abdulla et al.\ introduce \emph{view abstraction} for the verification of parameterised systems. Given a number $k \geq 1$, a \emph{view} of a word $w \in \Sigma^*$ is a scattered subword of $w$. Loosely speaking, Abdulla et al.\ abstract a word by its set of views of length up to $k$. In our setting, a constraint is a set  $F \subseteq \Sigma^{\leq k}$ of ``forbidden views'',  and $\interpretation(V)$ is the set of all words that do not contain any view of $F$. Since $k$ is fixed, this interpretation is regular.
\end{example}

\subsection{The abstract safety problem}
We apply regular abstraction frameworks to the problem of deciding whether an RTS avoids some regular set of unsafe configurations. 
For simplicity, we assume w.l.o.g.\ that the set of configurations of the RTS is $\Sigma^*$\footnote{By interpreting $\trafun$ as a relation over $\Sigma \times \Sigma$, any RTS can be transformed into an equivalent one with the same transitions where the set of configurations is $\Sigma^*$.}. Let us first formalise the \textsc{Safety} problem:

\begin{center}
\begin{minipage}{\textwidth}
Given:\ \begin{minipage}[t]{13cm} 
 a nondeterministic transducer recognising a regular relation $\trafun \subseteq \Sigma^* \times \Sigma^*$, and \\ two NFAs
 recognising regular sets $C_I, C_U \subseteq \Sigma^*$ of initial and unsafe configurations, respectively.
 \end{minipage} \\[0.2cm]
Decide:   does $\reach(C_I) \cap C_U = \emptyset$ hold?
\end{minipage}
\end{center}

It is a folklore result that \textsc{Safety} is undecidable. Let us sketch the argument. The configurations of a given Turing machine can be encoded as words of the form $w_l \, q \, w_r$, where $w_l, w_r$ encode the contents of the tape to the left and to the right of the head, and $q$ encodes the current state. With this encoding, the successor relation between configurations of the Turing machine is regular, and so is the set of accepting configurations. Taking the latter as set of unsafe configurations, the Turing machine accepts a given initial configuration if{}f the RTS started at the initial configuration is unsafe.

\subparagraph{\textsc{AbstractSafety}.} We show that regular abstraction framework induces an ``abstract'' version of the safety problem, in which we replace the reachability relation by an overapproximation derived from the abstraction framework. Fix an RTS $\system=(\configurations,\trafun)$ and a regular abstraction framework $\framework = (\configurations,\constraints, \interpretation)$. We introduce some definitions:

\begin{definition}
A set $C \subseteq \configurations$ of configurations is \emph{inductive} if $\trafun(C) \subseteq C$.  A constraint $\constraint$ is \emph{inductive} if $\interpretation(\constraint)$ is inductive. We let $\Inductive \subseteq \constraints$ denote the set of all inductive constraints of $\constraints$.  Given two configurations $c, c'$ and $\constraint \in \Inductive$, we say that $\constraint$ \emph{separates} $c$ from $c'$ if $c \in \interpretation(\constraint)$ and $c' \notin \interpretation(\constraint)$.
\end{definition}

\noindent It is a folklore result that $\reach(C)$ is the smallest inductive set containing $C$, and that if some $\constraint \in \Inductive$ separates $c$ and $c'$, then $(c, c') \notin \reach$. Hence, an abstraction framework $(\configurations, \constraints, \interpretation)$ induces a \emph{potential reachability} relation $\potreach \subseteq \configurations \times \configurations$,  defined as the set of all pairs of configurations that are not separated by any inductive constraint. Formally:
\begin{equation*}
\potreach  :=  \{ (c, c') \in \configurations \times \configurations \mid \forall \constraint \in \Inductive \colon c \in \interpretation(\constraint) \to c' \in \interpretation(\constraint) \} 
\end{equation*}

\noindent We have $\reach(C) \subseteq \potreach(C)$ for every set of configurations $C$.  In particular, given sets $C_I, C_U \subseteq \configurations$ of initial and unsafe configurations, if $\potreach(C_I) \cap C_U = \emptyset$, then the RTS is safe. 

\begin{example}
\label{ex:tokenpassingandtraps}
Consider the RTS of the token passing system of Example~\ref{ex:tokenpassing}, where $\Sigma=\{t,n\}$. We give two examples of abstraction frameworks. The first one is the abstraction framework of~\cite{ERW22}, already presented in Example~\ref{ex:traps}, with $b=1$. We have $\Gamma= 2^\Sigma = \{ \emptyset, \{t\}, \{n\}, \Sigma \}$. A constraint like $\varphi = \bigvee_{i=3}^5 \vari{t}{i}{5}$ is encoded by the word $\emptyset\emptyset\{t\}\{t\}\{t\} \in \Gamma^*$, and interpreted as the set of all configurations of length 5 that have a token at positions 3, 4, or 5, plus the set of all configurations of length different from 5. The two-state transducer for this interpretation is on the left of Figure~\ref{fig:interp}. For example, the left state has transitions leading to itself for the letters $\vtuple{\emptyset}{n}, \vtuple{\emptyset}{t}, \vtuple{\{t\}}{n},\vtuple{\{n\}}{t}$.
The constraint $\varphi$ is inductive. In fact, the language of all non-trivial inductive constraints (a constraint is trivial if it is satisfied by all configurations or by none) is $\{n\}^+\emptyset^*\{t\}^* + \{n\}^*\emptyset^*\{t\}^+$.
The set of configurations potentially reachable from $C_I = t n^*$ is $\potreach(C_I)=(tn + nn^*t)(t+n)^*$. In particular, $\potreach(C_I) \cap n^* = \emptyset$, but $tnt \in \potreach(C_I)$.  So this abstraction framework is strong enough to prove that every reachable configuration has at least one token, but not to prove that it has exactly one.

\begin{figure}[t]
\begin{center}
  \begin{minipage}[c]{0.5\textwidth}
    \begin{center}
     \resizebox{!}{20mm}{%
      \begin{tikzpicture}[]
        \node[state, initial, initial left, initial text=] (0) {};
        \node[state, right=3.0cm of 0, accepting] (1) {};
        \draw[->] (0) to node[above] {$\left \{ \vtuple{\gamma}{\sigma} \mid \sigma \in \gamma \right\}$} (1);
         \draw[->] (0) to node[below] {$\vtuple{\star}{\#}, \vtuple{\#}{\star}$} (1);
        \draw[->, loop above] (0) to node[above] {$\left \{ \vtuple{\gamma}{\sigma} \mid \sigma \notin \gamma \right\}$} (0);
        \draw[->, loop above] (1) to node[above] {$\Gamma_\# \times \Sigma_\#$} (1);
      \end{tikzpicture}
      }
    \end{center}
  \end{minipage}%
  \begin{minipage}[c]{0.5\textwidth}
    \begin{center}
    \resizebox{!}{3cm}{%
      \begin{tikzpicture}[]
        \node[state, initial, initial left, initial text=] (0) {};
        \node[state, right=40mm of 0, accepting] (1) {};
        \path (0) -- (1) coordinate[midway] (10) {};
        \node[state, below=1cm of 10, accepting] (2) {};
        \draw[->] (0) to node[above] {$\left \{ \vtuple{\gamma}{\sigma} \mid \sigma \in \gamma \right\}$} (1);
        \draw[->, loop above] (0) to node[above] {$\left \{ \vtuple{\gamma}{\sigma} \mid \sigma \notin \gamma \right\}$} (0);
        \draw[->, loop above] (1) to node[above] {$\left \{ \vtuple{\gamma}{\sigma} \mid \sigma \notin \gamma \right\}$} (1);
        \draw[->] (0) to node[left, yshift=-2mm] {$\vtuple{\star}{\#}, \vtuple{\#}{\star} \; $} (2);
        \draw[->] (1) to node[right, yshift=-2mm] {$\; \vtuple{\star}{\#}, \vtuple{\#}{\star}$} (2);
        \draw[->, loop right] (2) to node[right, yshift=-1mm] {$\Gamma_\# \times \Sigma_\#$} (1);
      \end{tikzpicture}
      }
    \end{center}
  \end{minipage}
\end{center}
\caption{Transducers for the interpretations of Example~\ref{ex:traps} and~\ref{ex:tokenpassingandtraps}. We have $\Gamma= 2^\Sigma$, and so the alphabet of the transducer is $(2^\Sigma)_\# \times \Sigma_\#$. The symbols $\vtuple{\star}{\#}$ 
and $\vtuple{\#}{\star}$ stand for the sets of all letters of the form  $\vtuple{\gamma}{\#}$ and $\vtuple{\#}{\sigma}$, respectively .}
\label{fig:interp}
\end{figure}
\end{example}

Consider now the framework in which, instead of a \emph{disjunction} of literals, a constraint is an \emph{exclusive disjunction} of literals, that is, a configuration satisfies the constraint if it satisfies \emph{exactly one} of its literals. So, in particular, the interpretation of $\emptyset\emptyset\{t\}\{t\}\{t\}$ is now that exactly one of the positions 3, 4, and 5 has a token. The interpretation is also regular; it is given by the three-state transducer on the right of Figure~\ref{fig:interp}. Examples of inductive constraints are $\{t\} \emptyset \{t, n\} \{n\}$ and all words of $\{t\}^*$. The language of non-trivial inductive constraints is given by the DFA on the left of Figure~\ref{fig:tokenpassingflows}. Observe that the set of words satisfying all constraints of $\{t\}^*$ is the language $n^*tn^*$. In particular, we have $\potreach(C_I) \subseteq n^*tn^* = \reach(C_I)$, and so $\potreach(C_I)=\reach(C_I)$.  \qed

\begin{figure}[t]
\begin{center}
\begin{minipage}[c]{0.9\textwidth}
    \resizebox{!}{4cm}{%
       \begin{tikzpicture}
        \node[state, initial, initial above, initial text=] (0) {};
        \node[state, right =1.5cm of 0, accepting] (2) {};
        \node[state, left =2.9cm  of 0] (1) {};
        \path (0) -- (1) coordinate[midway] (01);
        \node[state, below=22mm  of 0, accepting] (3) {};
        \node[state, right=1.5cm of 2] (7) {};
        \node[state, below=7mm of 7, accepting] (N) {};
        \node[state, below=22mm of 1, accepting] (5) {};
        \node[state, below=7mm of 2, accepting] (8) {};
        \node[state, right =1.5cm of 3, accepting] (9) {};    
        \node[state] (4) at (01 |- 8) {};  
        \draw[->] 
        (0) edge node[above] {$\{t\}$} (2)
        (0) edge node[above] {$\emptyset$} (1)
        (0) edge node[right,pos=0.2] {$\{n\}$} (3)
        (0) edge node[left,pos=0.15] {$\, \{t,n\}$} (4)
        (1) edge node[left] {$\{t\}$} (5)
        (1) edge node[right] {$\{t,n\}$} (4)
        (2) edge node[above] {$\emptyset$} (7)
        (2) edge node[left,pos=0.5] {$\{n\}$} (8)
        (2) edge node[right,pos=0.2] {$\{t,n\} \;$} (N)
        (3) edge node[below] {$\{t\}$} (5)
        (3) edge node[above left=-1mm] {$\{n\}$} (8)
        (3) edge node[below,pos=0.5] {$\emptyset \;$} (9)
        (4) edge node[below,pos=0.2] {$\{n\}$} (8)
        (7) edge node[right] {$\{t,n\}$} (N)
        (N) edge node[below] {$\; \{n\}$} (8)
        (1) edge[loop left] node {$\emptyset$} (1)
        (2) edge[loop above] node[right,xshift=2mm,yshift=-2mm] {$\{t\}$} (2)
        (4) edge[loop below] node[left,xshift=-2mm,yshift=1mm] {$\emptyset$} (4)
        (7) edge[loop right] node {$\emptyset$} (7)
        (8) edge[out=330,in=300,looseness=8] node[below right=-1mm] {$\{n\}$} (8)
        (9) edge[loop right] node {$\emptyset$} (9);      
      
    
        \node[state, right =1.7cm of N, yshift=1cm,initial, initial above, initial text=] (0p) {};
        \node[state, below =1.3cm of 0p, accepting] (2p) {};
        \draw[->] 
         (0p) edge node[right] {$\{t\}$} (2p)
        (2p) edge[loop right] node {$\{t\}$} (2p);
        \end{tikzpicture}
        }
\end{minipage}
\end{center}      
\caption{On the left, DFA recognising all non-trivial inductive constraints of Example~\ref{ex:tokenpassingflows}. On the right, fragment with the same interpretation as the DFA on the left.}
\label{fig:tokenpassingflows}
\end{figure}

\medskip The \textsc{AbstractSafety} problem is defined exactly as \textsc{Safety}, just replacing the reachability set $\reach(C_I)$ by the potential reachability set $\potreach(C_I)$ implicitly defined by the regular abstraction framework:

\begin{center}
\begin{minipage}{\textwidth}
Given:\ \begin{minipage}[t]{13cm} 
 a nondeterministic transducer recognising a regular relation $\trafun \subseteq \Sigma^* \times \Sigma^*$;
two NFAs recognising regular sets $C_I, C_U \subseteq \Sigma^*$ of initial and unsafe configurations, respectively; and
a deterministic transducer recognising a regular interpretation $\interpretation$ over $\Gamma \times \Sigma$.
 \end{minipage} \\[0.2cm]
Decide:  does $\potreach(C_I) \cap C_U = \emptyset$ hold?
\end{minipage}
\end{center}

Recall that \textsc{Safety} is undecidable. In the rest of this section and in the next one we show that \textsc{AbstractSafety} is \EXPSPACE-complete. Membership in \EXPSPACE\ was essentially proved in~\cite{ERW23}, while \EXPSPACE-hardness was left open. We briefly summarise the proof of membership in \EXPSPACE\ presented in~\cite{ERW23}, for future reference in our paper.

\begin{remark}
The result we prove in Section \ref{subsec:expspace} is slightly more general. In \cite{ERW23}, membership in \EXPSPACE\ is only proved for RTSs whose transducers are length-preserving, while we prove it in general. General transducers allow one to model parameterised systems with process creation. For example, we can model a token passing algorithm in which the size of the array can dynamically grow and shrink by adding the transitions  $\big(\vtuple{n}{n}+\vtuple{t}{t}\big)^{+} \; \big(\vtuple{n}{\#} + \vtuple{\#}{n}  \big)$ to the transition relation of Example~\ref{ex:tokenpassing}.
\end{remark}

\subsection{\textsc{AbstractSafety} is in \EXPSPACE}
\label{subsec:expspace}
We first show that the set of all inductive constraints of a regular abstraction framework is a regular language. Fix a regular abstraction framework $\framework=(\configurations, \constraints, \interpretation)$ over an RTS $(\configurations, \trafun)$. Let $n_\trafun$, $n_\interpretation$, $n_I$, $n_U$ be the number of states of the transducers and NFAs of a given instance of \textsc{AbstractSafety}.

\begin{lemma}{\cite{ERW23}}
\label{lem:inductive}
The set $\overline{\Inductive}$ is regular.  Further, one can compute an NFA with at most $n_\trafun \cdot n_{\interpretation}^2$ states recognising $\overline{\Inductive}$, and a DFA with at most  $2^{{n_\trafun} \cdot n_{\interpretation}^2}$ states recognising $\Inductive$.
\end{lemma}
\begin{proof}By definition, we have
\begin{align*}
  \overline{\Inductive}  & =  \{\constraint \in\Gamma^{*}\mid \exists c , c' \in \configurations \colon c \in \interpretation(\constraint), c'  \in \trafun(c), \text{ and } c' \notin \interpretation(\constraint) \} \\
   & =   \{ \constraint \in \Gamma^{*}  \mid  \exists c , c' \in \configurations \colon  (A,c) \in \interpretation, (c,c') \in  \trafun \text{ and } (c', \constraint) \in \overline{\interpretation^{-1}} \}
 \end{align*}
Let $\mathit{Id}_\Gamma = \{ (A,A) \mid \constraint \in \Gamma^*\}$. We obtain
$\overline{\Inductive}  = \big( \; \big(  \interpretation \circ  \trafun \circ \overline{\interpretation^{-1}} \big) \cap \mathit{Id}_\Gamma \; \big)\mid_1$. By the results at the end of Section~\ref{sec:prelims}, $\overline{\Inductive}$ is recognised by a NFA with $n_\interpretation \cdot n_\trafun \cdot n_\interpretation = n_\trafun n_\interpretation^2$ states, and so $\Inductive$ is recognised by a DFA with $2^{n_\trafun \cdot n_{\interpretation}^2}$ states.
\end{proof}

\begin{lemma}{\cite{ERW23}}
\label{lem:indcons}
The potential reachability relation $\potreach$ is regular. Further, one can compute a nondeterministic transducer with at most $K:=n_\interpretation^2 \cdot 2^{n_\trafun \cdot n_\interpretation^2}$ states recognising $\overline{\potreach}$, and a deterministic transducer with at most  $2^K$ states recognising $\potreach$.
\end{lemma}
\begin{proof}
By definition, we have
  \begin{align*}
 \overline{\potreach}  & =   \set{(c,c') \in \configurations \times \configurations  \mid \exists \constraint \in \Inductive \colon c \in \interpretation(\constraint)  \text{ and }  c' \notin \interpretation(\constraint)} \\
    & =  \set{(c,c') \in \configurations \times \configurations  \mid \exists \constraint \in \Inductive \colon (c, A) \in \interpretation^{-1} \text{ and } (A, c') \in \comp{\interpretation}}   
\end{align*}
Let $\mathit{Id}_\Gamma = \{ (A,A) \mid \constraint \in \Gamma^*\}$. We obtain $\comp{\potreach} = \big({\interpretation}^{-1} \circ (\textit{Id}_\Gamma \cap \Inductive^2) \circ \comp{\interpretation}   \big)$. Apply now the results at the end of Section~\ref{sec:prelims} and Lemma~\ref{lem:inductive}. 
\end{proof}

 
\begin{restatable}{theorem}{inEXPSPACE}{\cite{ERW23}}\label{thm:abstractP}
\textsc{AbstractSafety} is in \EXPSPACE.
\end{restatable}
\begin{proof}
Immediate consequence of Lemma~\ref{lem:indcons}, see \appref{app:inEXPSPACE}.
\end{proof}

\section{\textsc{AbstractSafety} is EXPSPACE-hard}
\label{sec:expspacecomp}
In~\cite{ERW22} it was shown that \textsc{AbstractSafety} was \PSPACE-hard, and the paper left the question of closing the gap between the upper and lower bounds open. We first recall and slightly alter the \PSPACE-hardness proof of~\cite{ERW22}, and then present our techniques to extend it to \EXPSPACE-hardness.

The proof is by reduction from the problem of deciding whether a Turing machine $\TM$ of size $n$ does not accept when started on the empty tape of size $n$. (For technical reasons, we actually assume that the tape has $n-2$ cells.) Given $\TM$, we construct in polynomial time an RTS $\system$ and a set of initial configurations $C_I$ that, loosely speaking, satisfy the following two properties: the execution of $\system$ on an initial configuration simulates the run of $\mathcal{M}$ on the empty tape, and $\potreach(C_I)=\reach(C_I)$. We choose $C_U$ as the set of configurations of $\system$ in which $\TM$ ends up in the accepting state. Then $\system$ is safe iff $\TM$ does not accept.

\subparagraph{Turing machine preliminaries.} We assume that $\TM$ is a deterministic Turing machine with states $Q$, tape alphabet $\Tapealph$, initial state $q_0$ and accepting state $q_f$.

We represent a configuration of $\TM$ as a word $\btape\,\beta\,q\,\eta$ of length $n$, where  $\TM$ is in state $q$, the content of the tape is $\beta\,\eta\in\Tapealph^*$, and the head of $\TM$ is positioned at the first letter of $\eta$. The symbol $\btape$ serves as a separator between different configurations. The initial configuration is $\alpha_0:=\btape q_0 \bk^{n-2}$, where $\bk$ denotes the blank symbol of $\TM$; so the tape is initially empty.

We assume w.l.o.g.\ that the successor of a configuration in state $q_f$ is the configuration itself, so the run of $\TM$ can be encoded as an infinite word $\Run:= \alpha_0 \alpha_1 \cdots$ where $\alpha_i$ represents the $i$-th configuration of $\TM$. For convenience, we write $\Confalph := Q \cup \Tapealph \cup \{\btape\}$ for the set of symbols in $\Run$. It is easy to see that the symbol at position $i+n$ of $\Run$ is completely determined by the symbols at positions $i-1$ to $i+2$ and the transition relation of $\TM$. We let $\delta(x_1x_2x_3x_4)$ denote the symbol which ``should'' appear at position $i+n$ when the symbols at positions $i-1$ to $i+2$ are $x_1x_2x_3x_4$; in particular, $\delta(x_1 \btape x_2 x_4) = \btape$.

\subparagraph{Configurations of $\system$.}
We choose the set of configurations as $\configurations := \alpha_0 \btape (\Confalph \cup \{\square\})^*$, and the initial configurations as $C_I := \alpha_0 \btape \square^*$. Intuitively, the RTS starts with the representation of the initial configuration of $\TM$, followed by some number of blank cells $\square$. During its execution, the RTS will ``write” the run of $\TM$ into these blanks.

A configuration is unsafe if it contains some occurrence of $q_f$, the accepting state of $\TM$, so $C_U:=(\Confalph\cup \{\square\})^*\{q_f\}(\Confalph \cup \{\square\})^*$.

\subparagraph{Transitions.} For convenience, we will denote the $i$-th position of a word $w$ as $w(i)$ instead of $w_i$. Given a configuration $c$, the set $\trafun(c)$ contains one single configuration $c'$, defined as follows. Let $i$ be some position of $c$ such that $c_{i+n} = \square$. Then $c'$ coincides with $c$ everywhere except at position $i+n$, where instead $c'(i+n):=\delta(c(i-1) c(i) c(i+1) c(i+2))$. It is easy to see that $\trafun$ is a regular relation: The transducer nondeterministically guesses the position $i-1$, reads the next four symbols, say $x_1...x_4$, stores $\delta(x_1...x_4)$ in its state, moves to position $i+n$, checks if $c(i+n)=\square$ and writes $c'(i+n):=\delta(x_1...x_4)$. The transducer has $\O(n^2)$ states.

\medskip\noindent It follows from the definitions above that $\TM$ accepts the empty word if{}f $\system$ can reach $C_U$ from $C_I$, i.e.\ $\reach(C_I)\cap C_U\ne\emptyset$.
    
\subparagraph{Regular abstraction framework.} We define a regular abstraction framework $\framework = (\configurations, \constraints, \interpretation)$ of polynomial size such that $\potreach(C_I) = \reach(C_I)$. Hence, for every configuration $c\notin\reach(C_I)$, we must find an inductive constraint $\constraint\in\constraints$ which separates $C_I$ and $c$. (Note that $C_I$ contains exactly one configuration of length $|c|$.)

As the reachable configurations are precisely the prefixes of $\Run$ with some symbols replaced by $\square$, there is a position $i$ s.t.\ $c(i) \notin \{\square,\Run(i)\}$. Let us fix the smallest such $i$. As we noted above, $\Run(i)$ is determined entirely by $\Run(i-n-1)...\Run(i-n+2)$ via the mapping $\delta$. So the constraint “if $c(i-n-1)...c(i-n+2)=x_1...x_4$, then $c(i)\in\{\square,\delta(x_1...x_4)\}$” is inductive and separates $C_I$ and $c$.

Therefore, it is sufficient to define an abstraction framework in which every constraint of the above form can be expressed. This is relatively straightforward. We set $\constraints:=\square^*\Confalph^4\square^*\Confalph\square^*$. Given a constraint $\constraint=\square^ix_1...x_4\square^jx\square^k$, define $\interpretation(\constraint)$ as the set of all configurations $c$ s.t.\ $c(i+1)...c(i+4)=x_1...x_4$ implies $c(i+j+5)\in\{\square,x\}$. Clearly, $\interpretation$ is a regular relation which can be recognised by a transducer with $3$ states.

\begin{theorem}{\cite{ERW22}}\label{thm:safetyPSPACE}
The abstract safety problem is \PSPACE-hard, even for regular abstraction frameworks where the transducer for the interpretation has a constant number of states. 
\end{theorem}

\begin{figure}[t]
\newcommand{\colorone}{\textcolor{nicepink}}
\newcommand{\colortwo}{\textcolor{nicepink!75!black}}
\newcommand{\m}[1]{\makebox[5mm]{$#1$}}

\[ \begin{aligned} 00 \quad 000 \quad &\m{\btape^0} \m{q_0^0} \m{\square^0} \m{\square^0} \m{\square^0} \m{\square^0} \m{\square^0} \m{\square^0} \m{\square^0} \m{\square^0}\\
\xrightarrow[]{\mathit{mark}(2,1)} 0{\colorone{1}} \quad 000 \quad &\m{\btape^{\colorone{1}}} \m{q_0^0} \m{\square^{\colorone{1}}} \m{\square^0} \m{\square^{\colorone{1}}} \m{\square^0} \m{\square^{\colorone{1}}} \m{\square^0} \m{\square^{\colorone{1}}} \m{\square^0}\\
\xrightarrow[]{\mathit{mark}(3,0)} 01 \quad {\colorone{1}}00 \quad &\m{\btape^1} \m{q_0^{\colorone{1}}} \m{\square^{\colortwo{1}}} \m{\square^0} \m{\square^{\colortwo{1}}} \m{\square^{\colorone{1}}} \m{\square^1} \m{\square^{\colorone{1}}} \m{\square^{\colortwo{1}}} \m{\square^0}\\
\xrightarrow[]{\mathit{init}} 00 \quad 000 \quad &\m{\btape^0} \m{q_0^0} \m{\square^0} \m{{\colorone{\bk}}^0} \m{\square^0} \m{\square^0} \m{\square^0} \m{\square^0} \m{\square^0} \m{\square^0}\\
\xrightarrow[]{\mathit{mark}(2,0)} {\colorone{1}}0 \quad 000 \quad &\m{\btape^0} \m{q_0^{\colorone{1}}} \m{\square^0} \m{\bk^{\colorone{1}}} \m{\square^0} \m{\square^{\colorone{1}}} \m{\square^0} \m{\square^{\colorone{1}}} \m{\square^0} \m{\square^{\colorone{1}}}\\
\xrightarrow[]{\mathit{mark}(3,0)} 10 \quad {\colorone{1}}00 \quad &\m{\btape^0} \m{q_0^{\colortwo{1}}} \m{\square^{\colorone{1}}} \m{\bk^1} \m{\square^{\colorone{1}}} \m{\square^{\colortwo{1}}} \m{\square^0} \m{\square^{\colortwo{1}}} \m{\square^{\colorone{1}}} \m{\square^1}\\
\xrightarrow[]{\mathit{init}} 00 \quad 000 \quad &\m{\btape^0} \m{q_0^0} \m{\square^0} \m{\bk^0} \m{\square^0} \m{\square^0} \m{{\colorone{\btape}}^0} \m{\square^0} \m{\square^0} \m{\square^0}\\
\xrightarrow[]{\cdots} 00 \quad 000 \quad &\m{\btape^0} \m{q_0^0} \m{{\colorone{\bk}}^0} \m{\bk^0} \m{{\colorone{\bk}}^0} \m{{\colorone{\bk}}^0} \m{\btape^0} \m{\square^0} \m{\square^0} \m{\square^0}\\
\xrightarrow[]{\cdots} {\colorone{1}}0 \quad 00{\colorone{1}} \quad &\m{\btape^{\colorone{1}}} \m{q_0^{\colorone{1}}} \m{\bk^0} \m{\bk^{\colorone{1}}} \m{\bk^{\colorone{1}}} \m{\bk^{\colorone{1}}} \m{\btape^{\colorone{1}}} \m{\square^{\colorone{1}}} \m{\square^0} \m{\square^{\colorone{1}}}\\
\xrightarrow[]{\mathit{write}} 00 \quad 000 \quad &\m{\btape^0} \m{q_0^0} \m{\bk^0} \m{\bk^0} \m{\bk^0} \m{\bk^0} \m{\btape^0} \m{\square^0} \m{\colorone{q_1^{\textcolor{black}{1}}}} \m{\square^0}\\
\xrightarrow[]{\cdots} 0{\colorone{1}} \quad 0{\colorone{1}}0 \quad &\m{\btape^{\colorone{1}}} \m{q_0^0} \m{\bk^{\colorone{1}}} \m{\bk^{\colorone{1}}} \m{\bk^{\colorone{1}}} \m{\bk^{\colorone{1}}} \m{\btape^{\colorone{1}}} \m{\square^0} \m{q_1^{\colorone{1}}} \m{\square^{\colorone{1}}}\\
\xrightarrow[]{\mathit{write}} 00 \quad 000 \quad &\m{\btape^0} \m{q_0^0} \m{\bk^0} \m{\bk^0} \m{\bk^0} \m{\bk^0} \m{\btape^0} \m{{\colorone{x}}^0} \m{q_1^0} \m{\square^0} \end{aligned} \]
\caption{A sample run of the regular transition system described in Example~\ref{ex:rtspspace}. Here, $\mathit{mark}(x,y)$ means that the $y$-th bit of the prime number $x$ is changed to 1, and thus every position not equivalent to $y \pmod x$ is unmarked. Note that the first position of the TM part (the one with $\btape$) is position 0. 
We write $x^y$ instead of $[{y \atop x}]$. We highlight bits and symbols that were written to in pink (bits which are unmarked by the \textit{mark} transition, but were already unmarked, are drawn in darker pink).
}\label{fig:rtspspace}
\end{figure}
\subsection{From \PSPACE-hardness to \EXPSPACE-hardness}\label{ssec:from-pspace-to-expspace}
In order to prove \EXPSPACE-hardness, we start with a machine $\TM$ of size $n$ and run it on a tape with $2^n$ cells. However, if we proceed exactly as in the \PSPACE-hardness proof, we encounter two obstacles: (1) The length of $\alpha_0$ is $2^n$, so our definitions of $\configurations$ and $C_I$ require automata of exponential size. (2) The transducer for the transition relation $\trafun$ needs to “count” to $2^n$, as this is the distance between the corresponding symbols of $\alpha_i$ and $\alpha_{i+1}$. Again, this requires an exponential number of states.

Obstacle (1) will be easy to overcome. Essentially, instead of starting the RTS with the entire initial configuration $\alpha_0$ of $\TM$ already in place, we set $C_I:=\btape\,q_0\,\square^*$ and modify the transitions of $\system$ to also write out $\alpha_0$.

However, obstacle (2) poses a more fundamental problem. On its face, it is easy to construct an RTS that can count to $2^n$ by executing multiple transitions in sequence, e.g.\ by implementing a binary counter. However, we need to balance this with the needs of the abstraction framework: if the RTS is too sophisticated, our constraints can no longer capture its behaviour using only regular languages.

We now sketch an RTS $\system'$ which extends the RTS $\system$ from the \PSPACE-hardness proof.

\subparagraph{A two-phase system.} In order to write the run of $\TM$, the RTS $\system'$ uses a “mark and write” approach. In a first phase, it executes $n$ transitions to mark positions with distance $\productofprimes$, where $\productofprimes\ge2^n$ is some fixed constant. Then, it nondeterministically guesses a marked position, reads and stores $4$ symbols from that position, and moves to the next marker to write the symbol according to $\delta$.

Let $p_1, \ldots, p_n$ be the first $n$ prime numbers (i.e.\ $p_1=2, p_2=3$, etc.). Define $\productofprimes:= \prod_{j=1}^n p_j$ and $s:= \sum_{j=1}^n p_i$. We have $\productofprimes \geq 2^n$ and, by the Prime Number Theorem, $s \in O(n^2 \log n)$.

The configurations of $\system'$ are of the form $\Markhead\,[{\Markers\atop c}]$, where $\Markhead\in\{0,1\}^s$ stores the current state of the mark phase, $\Markers\in\{0,1\}^*$ are the markers ($0$ means marked), and $c\in(\Confalph \cup \{\square\})^*$ is as for $\system$, with the reachable configurations being the prefixes of $\Run$ with some symbols replaced by $\square$. We refer to $[{\Markers\atop c}]$ as the \emph{TM part}.

The RTS has three kinds of transitions: $\trafun':=\trafun_{\mathit{mark}}\cup\trafun_{\mathit{write}}\cup\trafun_{\mathit{init}}$.

In the mark phase, $\system'$ executes a transition of $\trafun_{\mathit{mark}}$ for each $j\in[n]$. When executing such a transition, $\system'$ chooses a remainder $r \in [0,p_j-1]$ and sets the corresponding bit in $\Markhead$. It then unmarks every position in the TM part which is \emph{not} equivalent to $r$ modulo $p_j$ (by replacing the $0$ with a $1$). Hence, after executing $n$ transitions in $\trafun_\textit{mark}$, the positions of all 0's in the TM part are equivalent modulo every $p_j$. By the Chinese remainder theorem, these positions must also be equivalent modulo $\productofprimes$.

Afterwards, $\system'$ executes either a transition in $\trafun_{\mathit{write}}$ or $\trafun_{\mathit{init}}$. To execute $\trafun_{\mathit{write}}$, the RTS nondeterministically guesses a marked position $i$, reads $x:=c(i-1)...c(i+2)$, moves to the next marked position $i'$, and writes $\delta(x)$. 

As mentioned in obstacle (1) above, the RTS must write out the initial configuration of $\TM$. This is done by $\trafun_{\mathit{init}}$. If the first position of the TM part is not marked, the transducer moves to the first marked position and writes $\bk$, otherwise it moves to the second marked position and writes $\btape$. By executing this transition multiple times, eventually a configuration $\Markhead\,[{\Markers\atop c}]$ with $c=\btape q_0\bk^{\productofprimes-2}\btape\square^i$ can be reached.

While executing either $\trafun_{\mathit{write}}$ or $\trafun_{\mathit{init}}$, the transducer resets the mark phase state and marks all positions, i.e.\ the resulting configurations have $\Markhead=0^s$ and $\Markers \in 0^*$.

\begin{example}\label{ex:rtspspace}
Take $n = 2$. Here, we have $p_1 = 2$, $p_2 = 3$, $m = 6$ and $s = 5$. The set of initial and unsafe configurations is thus $C_I := L(0^5[{0\atop\btape}][{0\atop q_0}][{0\atop\square}]^*)$ and $C_U := \{0\}^5(\{0\}\times\Confalph)^*\{[{0\atop q_f}]\}(\{0\}\times\Confalph)^*$, respectively. In Figure~\ref{fig:rtspspace}, we give a possible run of the RTS for a TM with states $\{q_0,q_1,q_f\}$ ($q_0$ is initial, $q_f$ is final), and one transition from $q_0$, which reads $\bk$, moves the head to the right and goes to state $q_1$.
\end{example}

\subparagraph{The abstraction framework.}
If $\TM$ accepts, no constraint proving safety can exist, as an unsafe configuration is reachable. Consequently, when constructing the abstraction framework we only need to ensure that — provided $\TM$ does not accept — for every pair $(u,v)\in C_I\times C_U$ there is an inductive constraint separating $u$ and $v$.

The abstraction framework $(\configurations,\constraints,\interpretation)$ is the convolution of two independent parts, i.e.\ $\constraints:=\constraints_1\times\constraints_2$ and $\interpretation([{\constraint_1\atop\constraint_2}]):=\interpretation_1(\constraint_1)\cap\interpretation_2(\constraint_2)$.

\begin{figure}[t]
\newcommand{\m}[1]{\makebox[5mm][l]{$#1$}}

\begin{gather}
\begin{aligned}
\makebox[8mm][r]{$u=$}\quad\makebox[12mm]{$00 \quad 000$} \quad &\m{\btape^0} \m{q_0^0} \m{\square^0} \m{\square^0} \m{\square^0} \m{\square^0} \m{\square^0} \m{\square^0} \m{\square^0} \m{\square^0}\\
\makebox[8mm][r]{$v=$}\quad\makebox[12mm]{$00 \quad 000$} \quad &\m{\btape^0} \m{q_0^0} \m{\color{nicered}\bk^0} \m{\color{nicered}\bk^0} \m{\color{nicered}\bk^0} \m{\color{nicered}\bk^0} \m{\btape^0} \m{x^0} \m{q_1^0} \m{\color{nicered}q_f^0}\\
\end{aligned}\\[2pt]
\begin{aligned}
\makebox[8mm][r]{$\constraint_1=$}\quad\makebox[12mm]{$\square\;...\;\square$} \quad &\m{\square}\m{\square} \m{\bk}\m{\bk}\m{\bk}\m{\bk} \m{\square}\m{\square}\m{\square} \m{\bk}\\
\end{aligned}\\[2pt]
\begin{aligned}
\makebox[8mm][r]{$\constraint_2=$}\quad\makebox[12mm]{$01 \quad 100$}\, \quad &\m{0}\m{1}\m{0}\m{2}\m{0}\m{1}\m{0}\m{1}\m{0}\m{2}\\
\end{aligned}
\end{gather}
\caption{Constraints in Example~\ref{ex:rtspspace}. (1) Two configurations $u,v$, where $u\in C_I$, $v\in C_U$. (2) The (not necessarily inductive) constraint $\constraint_1$, separating $u,v$. (3) The matching inductive constraint $\constraint_2$. 
}\label{fig:rtspspace2}
\end{figure}

For every pair $(u,v)\in C_I\times C_U$ there will be a constraint $\constraint_1\in\constraints_1$ separating $u$ and $v$. This is similar to before: $v$ must contain an “error” somewhere, so our constraint will state “if $c(i-1)...c(i+2)=x$, then $c(i+\productofprimes)\in\{\delta(x),\square\}$”, for some $i,x$. (Depending on $v$ we instead may need $\constraint_1$ stating just “$c(i)\in\{\Run(i),\square\}$”.) Concretely, we set $\constraints_1:=\square^s\square^*(\Confalph^4\square^*\Confalph + \Confalph)\square^*$, so the constraint is represented by a word in $\square^*x\square^*\delta(x)\square^*$ (or a word in $\square^*\Run(i)\square^*$). An example is shown in Figure~\ref{fig:rtspspace2}.

This is enough to separate $u$ and $v$, as $v$ must contain an “error” somewhere (i.e.\ a deviation from $\Run$). But it is not inductive: We can take any configuration which has $c(i+\productofprimes)=\square$, but where the cells have not been marked correctly, s.t.\ executing $\trafun_{\mathit{write}}$ would write to position $i+\productofprimes$ after reading symbols $c(j-1)...c(j+2)$ with $j\ne i$. So the resulting configuration may have $c(i+\productofprimes)\ne\delta(x)$, which no longer fulfils $\constraint_1$.

We solve this issue via $\interpretation_2$. For the constraint $\constraint_1$ above there is going to be a constraint $\constraint_2\in\constraints_2$ s.t.\ the combination $\interpretation_1(\constraint_1)\cap\interpretation_2(\constraint_2)$ is inductive. Essentially, $\constraint_2$ will ensure that it is impossible to write to position $i+\productofprimes$ without reading from position $i$. (Note that for a particular constraint the value of $i$ is fixed.)

Let $\constraints_2:=\{0,1\}^s[0,n]^*$. Intuitively, a constraint $x\,y\in\constraints_2$ (where $|x|=s$) states: “if remainders for the first $j$ primes have been chosen according to $x$, then exactly the positions $k$ with $y(k)\ge j$ are marked, otherwise positions $k$ with $y(k)=n$ are unmarked”, where $j$ is the number of primes that have been chosen.

Again, the constraint $\constraint_1$ is only concerned with one position $i$. Moreover, there is only one sequence of remainders $r_1,...,r_n$ to choose for the $\trafun_{\mathit{mark}}$ transitions, s.t.\ position $i$ is marked (i.e.\ $r_j\equiv i\pmod{p_j}$). So for each position $k$ we can determine the index in the sequence of $\trafun_{\mathit{mark}}$ transitions at which position $k$ will first become unmarked. Concretely, we have $y(k):=\min\{j \mid k\not\equiv i\pmod{p_j}\} - 1$.

This constraint is inductive and, crucially, the intersection of $\constraint_1$ and $\constraint_2$ is inductive as well. Essentially, every $\trafun_{\mathit{mark}}$ transition either continues the sequence $r_1,...,r_n$, and then the positions must be marked precisely according to $y$, or at some point a different remainder has been chosen, and the position $i$ is unmarked and cannot be written to. 

To summarise, constraint $\constraint_1$ is sufficient to exclude any unsafe configuration and, in combination with $\constraint_2$, does so inductively. Therefore, if $\TM$ does not accept, then the RTS can be proven safe using the abstraction framework.

For the full proof, see \appref{app:expspace-hardness}.

\section{Learning regular sets of inductive constraints}
\newcommand{\subind}{\lang{H}}
\newcommand{\indhyp}{\lang{H}}

Recall the algorithm for \textsc{AbstractSafety} underlying Theorem \ref{thm:abstractP}. It computes an automaton recognising the set  $\Inductive$ of inductive constraints (Lemma \ref{lem:inductive}); uses this automaton to compute a transducer recognising the potential reachability relation $\potreach$ (Lemma \ref{lem:indcons}); uses  this transducer to compute  an automaton recognising $\potreach(C_I) \cap C_U$; and finally uses this automaton to check if $\potreach(C_I) \cap C_U$ is empty (Theorem \ref{thm:abstractP}). The main practical problem of this approach is that, while the automaton for $\overline{\Inductive}$ has polynomial size in the input, the automaton for $\Inductive$ can be exponential, and, while the automaton for $\overline{\potreach}$ has polynomial size in $\Inductive$, the size of the automaton for $\potreach$ can be exponential.

In practice one typically does not need all inductive constraints to prove safety. This can be illustrated even on the tiny RTS of Example \ref{ex:tokenpassing}.

\begin{example}
\label{ex:tokenpassingflows}
Consider the RTS of the token passing system of Example \ref{ex:tokenpassing}, where $\Sigma=\{t,n\}$,
and the second abstraction framework of Example \ref{ex:tokenpassingandtraps}, where $\Gamma= 2^\Sigma = \{ \emptyset, \{t\}, \{n\}, \{t,n\} \}$. Recall that in this abstraction framework a constraint is an \emph{exclusive disjunction} of literals, that is, a configuration satisfies the constraint if it satisfies \emph{exactly one} of its literals. The minimal DFA recognising all non-trivial inductive constraints was shown on the left of Figure \ref{fig:tokenpassingflows}. The set of inductive constraints $\{t\}\{t\}^*$ is satisfied by the configurations $n^*tn^*$, and so the DFA on the right is already strong enough to prove any safety property.
\end{example}

We present a learning algorithm that computes automata recognising increasingly large sets $\subind \subseteq \Inductive$ of inductive constraints until either $\subind$ is large enough to prove safety, or it becomes clear that even the whole set $\Inductive$ is not large enough. More precisely, recall that, by definition, we have $\potreach := \{ (c, c') \in \configurations \times \configurations \mid \forall \constraint \in \Inductive \colon c \in \interpretation(\constraint) \to c' \in \interpretation(\constraint) \}$. Given a set $\subind \subseteq \Inductive$, define the relation $\potreach_\subind$ exactly as $\potreach$, just replacing $\Inductive$ by $\subind$. Clearly, we have $\potreach_\subind \supseteq  \potreach$  and $\potreach_\Inductive = \potreach$.

\subsection{The learning algorithm}

Let $\system = (\configurations, \trafun)$ and $\framework=(\configurations, \constraints, \interpretation)$ be a regular transition system and a regular abstraction framework, respectively. Further, let $C_I$ and $C_U$ be regular sets of initial and unsafe configurations. The algorithm refines Angluin's algorithm $L^*$  for learning a DFA for the full set $\Inductive$ \cite{Angluin87,SteffenHM11}. Recall that Angluin's algorithm involves two agents, usually called  Learner and Teacher. Learner sends Teacher membership and equivalence queries, which are answered by Teacher according to the following specification:
\medskip

\begin{minipage}[t]{6cm}
\noindent \textbf{Membership Query}:
 \begin{itemize}
 \item Input: a constraint $\constraint \in \constraints$
 \item Output: $\checkmark$ if $\constraint \in \Inductive$, \\
 and $\times$ otherwise.
\end{itemize}
\end{minipage}
\begin{minipage}[t]{7cm}
\noindent \textbf{Equivalence Query}:
 \begin{itemize}
 \item Input: a DFA recognising a set $\lang{H} \subseteq \constraints$.
 \item Output: $\checkmark$ if $\lang{H} = \Inductive$, otherwise a constraint $\constraint \in (\lang{H} \setminus \Inductive) \cup (\Inductive \setminus \lang{H})$.
\end{itemize}
\end{minipage}

\medskip

\noindent  Angluin's algorithm describes a strategy for Learner guaranteeing that Learner eventually learns the minimal DFA recognising  $\Inductive$. The number of equivalence queries asked by Learner is at most the number of states of the DFA.

\subparagraph{Answering the queries.}  We describe the algorithms used by Teacher to answer queries. For membership queries, Teacher constructs an NFA for $\overline{\Inductive} \cap \{A\}$ with $O(|A| \cdot n_\trafun \cdot n_\interpretation^2)$ states (see Lemma \ref{lem:inductive}), and checks it for emptiness. 

For equivalence queries, Teacher proceeds as follows :
\begin{enumerate}
\item Teacher first checks whether $\lang{H}  \setminus \Inductive \neq \emptyset$ holds by computing an NFA recognising $\lang{H} \cap \overline{\Inductive}$ with $O(n_\lang{H} \cdot n_\trafun \cdot n_\interpretation^2)$ states (see Lemma \ref{lem:inductive}), and checking it for emptiness. If $\lang{H}  \setminus \Inductive$ is nonempty, then Teacher returns one of its elements.

\item Otherwise, Teacher constructs an automaton for $\potreach_{\lang{H}}(C_I) \cap C_U$ of size $O(2^{n_\interpretation^2 \cdot n_\lang{H}})$ and checks it for emptiness.  There are two cases:
\begin{enumerate}
\item If $\potreach_{\lang{H}}(C_I) \cap C_U = \emptyset$, then the system is safe; Teacher reports it and terminates.
In this case, the learning algorithm is aborted without having learned a DFA for $\Inductive$, because it is no longer necessary.
\item Otherwise, Teacher chooses an element $(c,c') \in \potreach_{\lang{H}} \cap (C_I \times C_U)$, and searches for an inductive constraint $\constraint$ such that $c \in \interpretation(\constraint)$ and $c' \notin \interpretation(A)$. We call this problem the \emph{separability problem}, and analyze it further in Section~\ref{sec:separation-problem}.
\end{enumerate}
\end{enumerate}

\subsection{The separability problem.}
\label{sec:separation-problem}
The \textsc{Separability} problem is formally defined as follows:

\begin{center}
\begin{minipage}{\textwidth}
Given:\ \begin{minipage}[t]{13cm} 
a nondeterministic transducer recognising a regular relation $\trafun \subseteq \Sigma^* \times \Sigma^*$;
a deterministic transducer recognising a regular interpretation $\interpretation$ over $\Gamma \times \Sigma$; and \\
two configurations $c, c' \in \configurations$
 \end{minipage} \\[0.2cm]
Decide:  is $c'$ separable from $c$, i.e.\ does there exist $\constraint \in \Inductive$ s.t.\ $c \in \interpretation(\constraint)$ and $c' \notin \interpretation(\constraint)$?
\end{minipage}
\end{center}

Contrary to \textsc{AbstractSafety},  the complexity of \textsc{Separability} is different for arbitrary transducers, and for length-preserving ones.

\begin{restatable}{theorem}{thmseparability}
\label{thm:sep}
\textsc{Separability} is \PSPACE-complete, even if $\trafun$ is length-preserving. If $\interpretation$ is length-preserving, then \textsc{Separability} is \NP-complete.
\end{restatable}
\begin{proof}
See \appref{app:separability-hardness}.
\end{proof}

Most applications of regular model checking to the verification of parameterised systems, and in particular all the examples studied in~\cite{ERW22,ERW23}, have length-preserving transition functions and length-preserving interpretations. For this reason, in our implementation we only consider this case, and leave an extension for future research. Since \textsc{Separability} is  \NP-complete in the length-preserving case, it is natural to solve it by reduction to SAT. 
A brief description of the reduction is given in \appref{app:reductiontoSAT}.

\subsection{Some experimental results}
We have implemented the learning algorithm in a tool prototype, built on top of
the libraries \texttt{automatalib} and \texttt{learnlib}~\cite{IsbernerHS15} and the SAT solver \texttt{sat4j}~\cite{BerreP10}.
We compare our learning approach with the one of~\cite{ERW22}, which constructs automata for $\Inductive$ and $\potreach$ using the regular abstraction framework of Example~\ref{ex:traps}.
In the rest of this section we call these two approaches the \emph{lazy} and the \emph{direct} approach, respectively. We use the same case studies as~\cite{ERW22}. 
We compare the sizes of the DFA for the final hypothesis $\lang{H}$ and $\potreach_\lang{H}$ with the sizes of the DFA for $\Inductive$ and $\potreach$.
The results are available at~\cite{welzel_mohr_2024_12734991} and are shown in Table~\ref{fig:results}.

\begin{table}[t]
  \begin{center}
    \scalebox{0.7}{
    \begin{tabular}[t]{l|R{0.2cm}R{0.3cm}|l|R{0.3cm}R{0.6cm}|R{0.5cm}R{0.5cm}|c}
      \multicolumn{4}{c}{}& \multicolumn{2}{c}{Lazy} & \multicolumn{2}{c}{Direct}& \\
      System & $\size{C_{I}}$ & $\size{\trafun}$ & Pr. & $\size{\lang{H}}$ & $\size{\textit{PR}_\lang{H}}$ &  $\size{\Inductive}$ &  $\size{\textit{PR}}$ &  Re.\\
      \hline
      \multirow{2}{*}{Bakery} &
      \multirow{2}{*}{$3$} & \multirow{2}{*}{$5$} & 
         D & $1$ & $1$   &  \multirow{2}{*}{$9$} & \multirow{2}{*}{$8$} & $\checkmark$ \\
        &&&M & $4$ & $3$   &  & &  $\checkmark$ \\
      \hline
      \multirow{2}{*}{Burns} &
      \multirow{2}{*}{$1$} & \multirow{2}{*}{$6$} &  D & $1$ & $1$   &  \multirow{2}{*}{$10$}& \multirow{2}{*}{$6$} & $\checkmark$  \\
        &&&M & $5$ & $3$   &  & & $\checkmark$ \\
      \hline
      \multirow{2}{*}{Dijkstra} &
      \multirow{2}{*}{$2$} & \multirow{2}{*}{$17$} &
         D & $4$ & $4$   &  \multirow{2}{*}{$218$} & \multirow{2}{*}{$22$}  & $\checkmark$  \\
        &&&M & $11$ & $8$   & & & $\checkmark$ \\
      \hline
      Dijkstra &
      \multirow{2}{*}{$2$} & \multirow{2}{*}{$12$} &
         D & $9$ & $9$   &  \multirow{2}{*}{$47$} & \multirow{2}{*}{$17$} & $\checkmark$ \\
       (ring)   &&&M & $9$ & $7$   &  & & $\times$ \\
      \hline
      Dining &
      \multirow{2}{*}{$2$} & \multirow{2}{*}{$8$} & \multirow{2}{*}{D} & \multirow{2}{*}{$23$} &  \multirow{2}{*}{$18$} & \multirow{2}{*}{$86$} & \multirow{2}{*}{$19$}  &\multirow{2}{*}{$\nicefrac{2}{2}$} \\
        crypto. &&&& & & & &  \\
        \hline
      \multirow{2}{*}{Herman} &
      \multirow{2}{*}{$2$} & \multirow{2}{*}{$11$} &
         D & $3$ & $2$   &  \multirow{2}{*}{$8$} & \multirow{2}{*}{$7$}   &$\checkmark$\\
        &&& O & $1$ & $2$  & & & $\checkmark$ \\
      \hline
      Herman &
      \multirow{2}{*}{$2$} & \multirow{2}{*}{$3$} &
         D & $1$ & $2$   &  \multirow{2}{*}{$7$} & \multirow{2}{*}{$7$}   &$\times$\\
       (linear)  &&& O & $1$ & $2$   & && $\checkmark$\\   
      \hline
      \multirow{2}{*}{Israeli-Jafon} &
      \multirow{2}{*}{$3$} & \multirow{2}{*}{$10$} &
         D & $1$ & $4$    &  \multirow{2}{*}{$21$} & \multirow{2}{*}{$7$}   &$\checkmark$\\
        &&& O & $1$ & $4$   &  && $\checkmark$\\
      \hline
      \multirow{1}{*}{Token passing} &
      \multirow{1}{*}{$2$} & \multirow{1}{*}{$3$} &
        O & $4$ & $4$   & \multirow{1}{*}{$9$} & \multirow{1}{*}{$7$}  &$\checkmark$\\
      \hline
      \multirow{1}{*}{Lehmann-Rabin} &
      \multirow{1}{*}{$1$} & \multirow{1}{*}{$7$} &
         D & $5$ & $6$   &  \multirow{1}{*}{$29$} & \multirow{1}{*}{$13$} &$\checkmark$\\
      \hline
      \multirow{1}{*}{LR phils. 1} &
      \multirow{1}{*}{$1$} & \multirow{1}{*}{$11$} &
         D & $13$ & $14$   &  \multirow{1}{*}{$29$}  & \multirow{1}{*}{$15$} &$\times$ \\
      \hline
      \multirow{1}{*}{LR phils. 2} &
      \multirow{1}{*}{$1$} & \multirow{1}{*}{$11$} &
         D & $25$ & $11$   &  \multirow{1}{*}{$29$} & \multirow{1}{*}{$9$} &$\checkmark$ \\
      \hline
      \multirow{1}{*}{Atomic phils.} &
      \multirow{1}{*}{$1$} & \multirow{1}{*}{$8$} &
         D & $13$ & $9$    &  \multirow{1}{*}{$22$} & \multirow{1}{*}{$20$} &$\checkmark$ \\
      \hline
      \multirow{2}{*}{Mux array} &
      \multirow{2}{*}{$2$} & \multirow{2}{*}{$4$} &
         D & $1$ & $2$   &  \multirow{2}{*}{$7$} & \multirow{2}{*}{$8$}   &$\checkmark$\\
        &&&M & $3$ & $5$   & & & $\times$ \\
        \hline
      \multirow{2}{*}{Res. alloc.} &
      \multirow{2}{*}{$1$} & \multirow{2}{*}{$5$} &
         D & $5$ & $5$   &  \multirow{2}{*}{$9$} & \multirow{2}{*}{$8$}  &$\checkmark$ \\
        &&&M & $3$ & $3$   & & & $\times$\\
        \hline
          \end{tabular}
          \quad
     \begin{tabular}[t]{l|R{0.2cm}R{0.3cm}|l|R{0.3cm}R{0.6cm}|R{0.5cm}R{0.5cm}|c}
      \multicolumn{4}{c}{}& \multicolumn{2}{c}{Lazy} & \multicolumn{2}{c}{Direct}& \\
      System & $\size{\mathcal{C}_{I}}$ & $\size{\trafun}$ & Pr. & $\size{\lang{H}}$ & $\size{\textit{PR}_\lang{H}}$ &  $\size{\Inductive}$ &  $\size{\textit{PR}}$ &  Re. \\
       
      \hline
      \multirow{2}{*}{Berkeley} &
      \multirow{2}{*}{$1$} & \multirow{2}{*}{$9$} &
         D & $1$ & $1$   &  \multirow{2}{*}{$12$} & \multirow{2}{*}{$9$} &$\checkmark$ \\
        &&&O & $4$ & $4$ &  && $\nicefrac{2}{3}$ \\
      \hline
      \multirow{2}{*}{Dragon} &
      \multirow{2}{*}{$1$} & \multirow{2}{*}{$23$} &
         D & $1$ & $1$   &  \multirow{2}{*}{$37$} & \multirow{2}{*}{$11$} &$\checkmark$\\
        &&&O & $15$ & $7$ & &&  $\nicefrac{6}{7}$\\
      \hline
      \multirow{2}{*}{Firefly} &
      \multirow{2}{*}{$1$} & \multirow{2}{*}{$16$} &
         D & $1$ & $1$   & \multirow{2}{*}{$12$} & \multirow{2}{*}{$7$} &$\checkmark$\\
        &&&O & $4$ & $3$ & & & $\nicefrac{0}{4}$\\
      \hline
      \multirow{2}{*}{Illinois} &
      \multirow{2}{*}{$1$} & \multirow{2}{*}{$16$} &
         D & $1$ & $1$    &  \multirow{2}{*}{$18$} & \multirow{2}{*}{$14$}   &$\checkmark$\\
        &&&O & $4$ & $3$ & &&  $\nicefrac{0}{2}$\\
      \hline
      \multirow{2}{*}{MESI} &
      \multirow{2}{*}{$1$} & \multirow{2}{*}{$7$} &
         D & $1$ & $1$   &  \multirow{2}{*}{$8$} & \multirow{2}{*}{$7$}   &$\checkmark$\\
        &&& O & $4$ & $4$  & & & $\nicefrac{2}{2}$\\
      \hline
      \multirow{2}{*}{MOESI} &
      \multirow{2}{*}{$1$} & \multirow{2}{*}{$7$} &
         D & $1$ & $1$    &  \multirow{2}{*}{$15$} & \multirow{2}{*}{$10$} &$\checkmark$\\
        &&&O & $4$ & $4$ & &&  $\nicefrac{7}{7}$\\
      \hline
      \multirow{2}{*}{Synapse} &
      \multirow{2}{*}{$1$} & \multirow{2}{*}{$5$} &
         D & $1$ & $1$   &  \multirow{2}{*}{$8$} & \multirow{2}{*}{$7$} &$\checkmark$ \\
        &&&O & $2$ & $3$ &  &&  $\nicefrac{2}{2}$\\
      \hline
      \multicolumn{9}{c}{}
      \\
      \multicolumn{9}{c}{}
      \\
      \multicolumn{4}{c}{}& \multicolumn{2}{c}{Lazy} & \multicolumn{2}{c}{Direct}& \\
      System & $\size{\mathcal{C}_{I}}$ & $\size{\trafun}$ & Pr. & $\size{\lang{H}}$ & $\size{\textit{PR}_\lang{H}}$ &  $\size{\Inductive}$ &  $\size{\textit{PR}}$ &  Re. \\
      \hline
      Dijkstra &
      \multirow{2}{*}{$2$} & \multirow{2}{*}{$12$} &
        \multirow{2}{*}{M} & \multirow{2}{*}{$9$} & \multirow{2}{*}{$7$} &  & & \multirow{2}{*}{$\checkmark$} \\
        (ring) &  & & & & & &  \\
      \hline
      \multirow{1}{*}{LR phils. 1} &
      \multirow{1}{*}{$1$} & \multirow{1}{*}{$11$} &
         D & $34$ & $11$   &  &  &$\checkmark$ \\
        \hline
        \multirow{1}{*}{Mux array} &
        \multirow{1}{*}{$2$} & \multirow{1}{*}{$4$} &
          M & $5$ & $3$  & & &$\checkmark$\\
        \hline
        \multirow{1}{*}{Res. alloc.} &
        \multirow{1}{*}{$1$} & \multirow{1}{*}{$5$} &
         M & $5$ & $5$   &   &  &$\checkmark$ \\
         \hline
    \end{tabular}
  }
  \end{center}
 \caption{Comparison of the sizes of the automata computed by the lazy and direct approaches. In each table, the first three columns contain the name of the RTS and the sizes of the automata for $C_I$ and $\trafun$.
 The fourth column (Pr.) indicates the checked property, where D, M, and O stand for ``deadlock freedom'', ``mutex'' (at most one process in a given state), and ``other'' (custom properties of the particular RTS).
The next two columns give the results for the lazy approach: sizes of the DFAs for $\lang{H}$ and $\potreach_\lang{H}$ (abbreviated as $\textit{PR}_\lang{H}$), and the next two the same results for the direct approach.
The last column (Re.) indicates the result of the check: the property could be proved ($\checkmark$), could not ($\times$), or, in the case of multiple properties, how many of the properties were proved.}
  \label{fig:results}
\end{table}

The left table in Table~\ref{fig:results} shows results on RTSs modeling mutex and leader election algorithms, and academic examples, like various versions of the dining philosophers.
The right top table shows results on models of cache-coherence protocols.
Observe that $\Inductive$ and $\potreach$ do not depend on the property, but $\lang{H}$ and $\potreach_\lang{H}$ do, because the algorithm can finish early. In this case, the sizes given in columns $\lang{H}$ and $\potreach_\lang{H}$ are the largest
ones computed over all properties checked.

The main result is that the automata computed by our tool are significantly smaller than those for~\cite{ERW22}.
(Note that in all cases we compute minimal DFAs, and so the differences are not due to algorithms for the computation of automata.)
Observe that in five cases the deadlock-freedom and the mutex properties could not be proved. In one case (deadlock-freedom of Herman (linear)) this is because the property does not hold. 
In the other four cases, the problem is that~\cite{ERW22} uses only a specific regular abstraction framework, namely the one of Example~\ref{ex:traps}. We can prove the property
by refining the abstraction: we take the union of the ``disjunctive'' and the ``exclusive disjunctive'' abstractions of Example~\ref{ex:tokenpassingandtraps}. 
The  bottom-right table gives the results of these four cases. 

Both tools take less than three seconds in 54 out of the 59 case studies in the left and top right tables. We do not report the exact times; the implementation of~\cite{ERW22} uses MONA, while the experiments of this paper use \texttt{automatalib} and \texttt{learnlib}, and so small time differences may have any number of reasons. In the other five cases, the implementation of~\cite{ERW22} still needs less than one second, while our implementation takes minutes (more than ten minutes in two cases).  In these five cases the time performance is dominated by the SAT solver \texttt{sat4j}. We have not yet identified a pattern explaining why \texttt{sat4j} takes so much time, in particular the number and size of the formulas passed to it is similar to the other cases.

\section{Conclusions}

We have generalised the technique of~\cite{ERW22,ERW23} for checking safety properties of RTS to arbitrary regular abstraction frameworks. 
We have shown that the abstract safety problem is \EXPSPACE-complete, solving an open problem of \cite{ERW22,ERW23}, by means of a complex reduction of independent interest. For particular abstraction frameworks the complexity can be better.

We have used automata learning to design a lazy algorithm that stops when the inductive constraints computed so far are enough to prove safety. Its combination with other learning techniques, as those proposed in~\cite{Neider14,AbhayPhD,ChenHLR17,VardhanSVA04,NeiderJ13}, is a question for future research.

\bibliography{references.bib}

\appendix

\section{Operations on regular abstraction frameworks}\label{app:operations}

\operations*

\begin{proof}
For $i \in \{1,2\}$, let $(Q_i, \Gamma_i \times \Sigma, \delta_i, q_{0i}, F_i)$ be a deterministic transducer for $\interpretation_i$. By the product construction,
\[ (Q_1 \times Q_2, (\Gamma_1 \times \Gamma_2) \times \Sigma, \delta_\cap, (q_{01},q_{02}), F_1 \times F_2) \]
with
\[ \delta((q_1,q_2),\bigg[{(\gamma_1,\gamma_2) \atop a}\bigg]) = (\delta_1(q_1,a), \delta_2(q_2,a)) \]
is a deterministic transducer for $\interpretation_\cap$.

For union, assume that $\Gamma_1$ and $\Gamma_2$ are disjoint. In this case,
\[ (Q_1 \cup Q_2 \cup \{q_0, q_\emptyset\}, (\Gamma_1 \cup \Gamma_2) \times \Sigma, \delta_\cup, q_0, F_1 \cup F_2) \]
with
\[ \delta_\cup(q,a) := \begin{cases} \delta_i(q,a) &\text{ if } q \in Q_i \text{ and } a \in \Gamma_i \times \Sigma\\
\delta_i(q_{0i},a) &\text{ if } q = q_0 \text{ and } a \in \Gamma_i \times \Sigma\\
q_\emptyset &\text{ if } q = q_\emptyset \text{ and } a \in (\Gamma_1 \cup \Gamma_2) \times \Sigma \end{cases} \]
is a deterministic transducer for $\interpretation_\cup$.
\end{proof}

\section{\textsc{AbstractSafety} is in \EXPSPACE}\label{app:inEXPSPACE}

\inEXPSPACE*

\begin{proof}
By Lemma~\ref{lem:indcons}, one can effectively compute an NFA for $\potreach(C_I) = C_I \circ \potreach$ with $2^{K} \cdot n_I$ states, and so an NFA for $\potreach(C_I) \cap C_U$ with $2^K \cdot n_I \cdot n_U$ states. We give a nondeterministic algorithm for the abstract safety problem. The algorithm guesses step by step a run of this NFA leading to an accepting state. The memory required is the memory needed to store one state of the NFA. Since the DFA has $2^K \cdot n_I \cdot n_U$ states, the memory needed is $O(K +  \log n_I + \log n_U)$ which is exponential in the size of the input. So the abstract safety problem is in $\textsf{NEXPSPACE} = \EXPSPACE$.
\end{proof}

\section{\textsc{AbstractSafety} is \EXPSPACE-hard}\label{app:expspace-hardness}

We prove that \textsc{AbstractSafety} is \EXPSPACE-hard.
We briefly recall and make formally precise the definitions from Section~\ref{ssec:from-pspace-to-expspace}.

\subparagraph*{Turing machine.}
Let $\TM$ denote a deterministic Turing machine with tape alphabet $\Gamma'$, states $Q$, initial state $q_0$, final state $q_f$, and let $n:=|Q|$. We write $p_1,...,p_n$ for the first $n$ primes, and set $s:=\sum_ip_i$ and $\productofprimes:=\prod_ip_i$.

We reduce from the problem of deciding whether $\TM$ accepts when run on an empty tape of size $\productofprimes-2\ge 2^n$. (In the following, we shorten this to “$\TM$ accepts”.)

Configurations of $\TM$ are represented as described in Section~\ref{sec:expspacecomp}, so they are words of length $\productofprimes$ over the alphabet $\Confalph:=Q \cup \Tapealph \cup \{\btape\}$. We write $\alpha_i$ for the representation of the $i$-th configuration of $\TM$.

The infinite word $\Run:=\alpha_0\alpha_1...$ encodes the run of $\TM$ on an empty tape of size $\productofprimes-2$. From the transitions of $\TM$ we obtain a function $\delta:\Confalph^4\rightarrow\Confalph$ s.t.\ $\Run(i+\productofprimes)=\delta(\Run(i-1)...\Run(i+2))$ for all $i$.
We extend $\delta$ to a function $\delta:(\Confalph\cup\{\square\})^4\rightarrow\Confalph\cup\{\square\}$ by mapping every word $x\in(\Confalph\cup\{\square\})^4\setminus\Confalph^4$ to $\square$.

\begin{observation}\label{obs:tmaccepts}
$\TM$ accepts iff there exists a word $c\in\Confalph^*$ s.t.\ $c$ begins with $\alpha_0\btape$, $c(i+\productofprimes)=\delta(c(i-1)...c(i+2))$ for all $i \geq 2$, and $c$ contains $q_f$.
\end{observation}

\subparagraph*{Regular transition system.}
We define an RTS $\system'=(\configurations,\trafun)$. The configurations are $\configurations:=\{0,1\}^s(\{0,1\}\times\Confalph)^*$. For a configuration $u=\Markhead\,[{\Markers\atop c}]\in\configurations$ we refer to $\Markhead$ as the \emph{prime part} and $[{\Markers\atop c}]$ as the \emph{TM part}.

We say that prime $p_j$ is \emph{selected} in $u$ if $\Markhead(p_1+...+p_{j-1}+i)=1$ for some $i\in[p_j]$. We write $J(u):=\max(\{0\}\cup\{j:p_j\text{ is selected in }u\})$ for the index of the largest selected prime.

Within the TM part, we say that a symbol $[{0 \atop x}]$ is \emph{marked}, and $[{1 \atop x}]$ is \emph{unmarked}. One \emph{marks} (or \emph{unmarks}) a symbol $[{b \atop x}]$ by replacing it with $[{0 \atop x}]$ (or $[{1 \atop x}]$). To \emph{write} $y$ to $[{b \atop x}]$, we replace it with $[{b \atop y}]$ if $x=\square$, otherwise we leave it unchanged.

Let us now construct the transitions of $\system'$. As described in Section~\ref{ssec:from-pspace-to-expspace}, we set $\trafun':=\trafun_{\mathit{mark}}\cup\trafun_{\mathit{write}}\cup\trafun_{\mathit{init}}$.

We start with $\trafun_{\mathit{mark}}$, which is shown in Figure~\ref{fig:markandwrite}. The idea is that we pick some residual $r$ of some prime $p_j$, and unmark all positions that do \emph{not} have the same residual.

\newcommand{\Cgood}{C_\mathit{good}}
For $\trafun_{\mathit{write}}$, we split the description into two parts: a regular relation $\trafun_\mathit{write}'$ which implements the modification of the configuration, and a regular language $\Cgood$ describing invariants that should hold while executing $\trafun_\mathit{write}'$. The final relation is $\trafun_\mathit{write}:=\{(u,v)\in \trafun_\mathit{write}':u\in \Cgood\}$. In other words, the RTS only executes a write transition if some consistency conditions are met.

The transducer for $\trafun_\mathit{write}'$ is given in Figure~\ref{fig:markandwrite}. Informally, the language $\Cgood$ contains the configurations where the prime part has precisely one bit set for every prime. Writing $L_i\subseteq\{0,1\}^i$ for the binary words of length $i$ with exactly one $1$, we thus set $\Cgood:=L_{p_1}L_{p_2}...L_{p_n}(\{0,1\}\times\Confalph)^*$.

It remains to describe $\trafun_\mathit{init}$. Analogously to $\trafun_\mathit{write}$, we define a relation $\trafun_\mathit{init}'$ and only execute the transition if the input is in $\Cgood$. The relation is shown in Figure~\ref{fig:initandinterp}; it initialises the tape of $\TM$ by replacing the first marked position after $s+1$ with the appropriate symbol.

\begin{figure}
\hspace*{-0.05\textwidth}
\begin{minipage}{0.55\textwidth}
\begin{algorithm}[H]
\KwIn{$u\in \configurations$}
\KwOut{$v\in \configurations$ such that $[{u \atop v}] \in \trafun_\mathit{mark}$}
$v:=u$\;
$j:=J(u)+1$\;
pick $r\in[0,p_j-1]$\; \label{line:m-pick}
$v(p_1+...+p_{j-1}+1+r):=1$\;
\For{$k=1,2,...$}{
    \If{$k\not\equiv r\pmod{p_j}$}{
	   unmark $v(s+k)$\;
	}
}
\end{algorithm}
\end{minipage}
\hspace*{-0.05\textwidth}
\begin{minipage}{0.54\textwidth}
\begin{algorithm}[H]
\KwIn{$u\in \configurations$}
\KwOut{$v\in \configurations$ such that $[{u \atop v}] \in \trafun_\mathit{write}'$}
$v:=u$\;
$v(1)...v(s):=0^s$\;
pick $i\ge s+2$ s.t.\ $u(i)$ is marked\;
$[{m \atop x}]:=u(i-1)...u(i+2)$\; \label{line:w-read}
pick smallest $i'>i$, s.t.\ $u(i')$ is marked\; \label{line:w-pick2}
mark $v(s+1),v(s+2),...$\;
write $\delta(x)$ to $v(i')$\; \label{line:w-write}
\end{algorithm}
\end{minipage}
\caption{Pseudo-code representation of the transducers for $\trafun_\mathit{mark}$ and $\trafun_\mathit{write}'$.}\label{fig:markandwrite}
\end{figure}
Finally, we give the sets of initial and unsafe configurations: $C_I := L(0^s[{0\atop\btape}][{0\atop q_0}][{0\atop\square}]^*)$, $C_U:=0^s(\{0\}\times\Confalph)^*\{[{0\atop q_f}]\}(\{0\}\times\Confalph)^*$.

\begin{lemma}\label{lem:unsafeifaccepts}
If $\TM$ accepts, $\reach(C_I)\cap C_U\ne\emptyset$.
\end{lemma}
\begin{proof}
As $\TM$ accepts, there is some $l$ with $\alpha(l)=q_f$. We show that we can move from any configuration $u_i=0^s\,[{0^l\atop c}]\in\configurations$ where $c=\alpha(1)...\alpha(i)\square^{l-i}$ to the configuration $u_{i+1}$. In other words, we can write the run of $\TM$ down one symbol at a time, until we reach $q_f$.

Clearly, $u_2\in C_I$ and $u_l\in C_U$, so this suffices to show the lemma.

Fix some $i$, and let $r_1,...,r_n$ denote the unique remainders with $r_j\in[0,p_j-1]$ and $r_j\equiv i+1\pmod{p_j}$ for $j\in[n]$. Starting from $u_i$, we execute $\trafun_{\mathit{mark}}$ a total of $n$ times, selecting remainder $r:=r_j$ during the $j$-th execution in line~\ref{line:m-pick}. Let $u'$ denote the resulting execution. In $u'$, position $u'(s+i+1)$ is marked.

We now have two cases. If $i+1-\productofprimes>1$, then $u'(s+i+1-\productofprimes)$ is marked as well, and we execute $\trafun_{\mathit{write}}$ to write $\alpha(i+1)=\delta(\alpha(i-\productofprimes)...\alpha(i-\productofprimes+3))$ to $u'(s+i+1)$, which uses $u'(s+i-\productofprimes)...u'(s+i-\productofprimes+3)=[{b\atop \alpha(i-\productofprimes)...\alpha(i-\productofprimes+3)}]$. Otherwise, executing $\trafun_{\mathit{init}}$ writes $\alpha(i+1)=\alpha_0(i+1)\in\{\bk,\btape\}$ to $u'(s+i+1)$.

In either case, the prime part is reset and all positions in the TM part are marked, so the resulting configuration is precisely $u_{i+1}$.
\end{proof}

\subparagraph*{Regular abstraction framework.}
Lemma~\ref{lem:unsafeifaccepts} shows that, if $\TM$ accepts, no constraint proving safety can exist. Consequently, when constructing the abstraction framework we only need to ensure that — provided $\TM$ does not accept — for every pair $(u,v)\in C_I\times C_U$ there is an inductive constraint separating $u$ and $v$.

As described in Section~\ref{ssec:from-pspace-to-expspace}, the abstraction framework $(\configurations,\constraints,\interpretation)$ is the convolution of two independent parts, i.e.\ $\constraints:=\constraints_1\times\constraints_2$ and $\interpretation([{\constraint_1\atop\constraint_2}]):=\interpretation_1(\constraint_1)\cap\interpretation_2(\constraint_2)$, where the constraints in $\constraints_1$ are sufficient to separate each pair of initial and unsafe configuration, but are not inductive, and the constraints in $\constraints_2$ are added to guarantee inductiveness.

\newcommand{\ConfAlt}{\configurations_l}
\subparagraph*{Non-inductive separation: $\interpretation_1$.}
We set $\constraints_1:=\square^s\square^*(\Confalph^4\square^*\Confalph + \Confalph)\square^*$. Let $\constraint\in\constraints_1$ denote a constraint, and $\ConfAlt$ the configurations of length $|\constraint|$. For $x\in\Confalph^*$ and index $i$, we write $C_i^x\subseteq\ConfAlt$ for the set of configurations $\Markhead\,[{\Markers\atop c}]$ s.t.\ $c(i+1)...c(i+|x|)=[{b\atop x}]$ for some $b\in\{0,1\}^{|x|}$.

If $\constraint=\square^{s+i}x\square^jy\square^k$, we set $\interpretation_1(\constraint)$ to the configurations $\Markhead\,[{\Markers\atop c}]$ of length $|\constraint|$ where $c(i+j+5)\in\{y,\square\}$ or $c(i+1)...c(i+4)\in \Confalph^4\setminus\{x\}$.

If $\constraint=\square^{s+i}y\square^k$, we simply set $\interpretation_1(\constraint)$ to the configurations with $c(i+1)\in\{y,\square\}$.

\begin{lemma}\label{lem:interpretation1-separating}
Let $(u,v)\in C_I\times C_U$. If $\TM$ does not accept, there exists an $\constraint\in\constraints_1$ with $u\in\interpretation_1(\constraint)$ and $v\notin\interpretation_1(\constraint)$.
\end{lemma}
\begin{proof}
Let $\Markhead\,[{\Markers\atop c}]:=v$. We have $v\in C_U$ and thus $c\in\Confalph^*$. As $\TM$ does not accept, $c$ must violate one of the conditions of Observation~\ref{obs:tmaccepts}. We do a case distinction on which.

If $c$ does not start with $\alpha_0\btape$, there is some $i$ with $c(i)\ne(\alpha_0\btape)(i)=:x$ and we set $\constraint:=\square^{s+i-1}x\square^{|c|-s-i}$. As $u\in C_I$, we have $u(i)\in\{[{0\atop \square}],[{0\atop x}]\}$, so $\constraint$ separates $u$ and $v$.

If there is an $i$ with $c(i+\productofprimes)\ne \delta(c(i-1)...c(i+2))=:y$, let $x:=c(i-1)...c(i+2)$. We choose $\constraint:=\square^{s+i-2}x\square^{\productofprimes-3}y\square^{|c|-s-i-\productofprimes}$. Using $u(i)=[{0\atop \square}]$, we immediately get $u\in\interpretation(\constraint)$.
\end{proof}

\begin{figure}
\hspace*{-0.05\textwidth}
\begin{minipage}{0.5\textwidth}
\begin{algorithm}[H]
\KwIn{$u\in \configurations$}
\KwOut{$v\in \configurations$ such that $[{u \atop v}] \in \trafun_\mathit{init}'$}
$v:=u$\;
$v(1)...v(s):=0^s$\;
\lIf{$u(s+1)$ is marked}{$x:=\btape$}
\lElse{$x:=\bk$}
pick smallest $i'>s+1$ s.t.\ $u(i')$ is marked\;
mark $v(s+1),v(s+2),...$\;
write $x$ to $v(i')$\;
\end{algorithm}
\end{minipage}
\hspace*{-0.05\textwidth}
\begin{minipage}{0.59\textwidth}
\begin{algorithm}[H]
\KwIn{$u\in \configurations$, $\constraint\in\constraints_2$}
\KwOut{whether $u\in\interpretation_2(\constraint)$}
$j:= J(u)$, $l:=p_1+...+p_j$\;
\If{$u(k)\ne\constraint(k)$ for some $k\in[l]$}{
    \For{$i>s$ s.t.\ $\constraint(i)=n$}{
        \lIf{$u(i)$ marked}{reject}
    }
    accept\; \label{line:i-acc-1}
}
\For{$k=s+1,...$}{
    \lIf{$\constraint(k)\ge j\nleftrightarrow u(k)$ marked}{reject}
}
accept\; \label{line:i-acc-2}
\end{algorithm}
\end{minipage}
\caption{Pseudo-code representation of the transducers for $\trafun_\mathit{init}$ and $\interpretation_2$.}\label{fig:initandinterp}
\end{figure}

\subparagraph*{Adding inductiveness: $\interpretation_2$.}
Let $\constraints_2:=\{0,1\}^s[0,n]^*$. We say that prime $p_j$ is \emph{selected} in a configuration $\Markhead\,[{\Markers\atop c}]\in\configurations$ if $\Markhead(p_1+...+p_{j-1}+i)=1$ for some $i\in[p_j]$. Figure~\ref{fig:initandinterp} gives a pseudo-code representation of $\interpretation_2$. Note that we allow $j=0$ in the first line, for the case that no prime is selected.

\newcommand{\ConstraintIndex}[2]{\constraint_2^{(#1,#2)}}
Let $i$ denote an index. We will now define a constraint $\ConstraintIndex{i}{l}\in\constraints_2$, which essentially encodes that the marking procedure works correctly for indices $i,i+\productofprimes,i+2\productofprimes,...$ of the TM part, if it is of length $l$. Let $r_1,...,r_n \in [0,p_j-1]$ denote the unique remainders s.t.\ $r_j\equiv i\pmod{p_j}$. We will refer to $r_1,...,r_n$ as the \emph{remainder sequence} of $i$. Let $x\in\{0,1\}^s$ be a prime part where exactly the bits corresponding to $r_1,...,r_n$ have been set, i.e.\ $x(p_1+...+p_{j-1}+1+r_j)=1$ for each $j$, and $x$ is $0$ elsewhere. Let $y\in[0,n]^l$ s.t.\ $y(k):=\min\{j \mid k\not\equiv i\pmod{p_j}\} - 1$ and $y(k) := n$ if $k \equiv i \pmod{p_j}$ for all $j \in [n]$. Then $\ConstraintIndex{i}{l}:=x\,y$.

We remark that positions $k$ with $\ConstraintIndex{i}{l}(k)=j$ have distance at least $p_1\cdot...\cdot p_j$; in particular, $...,y(i-\productofprimes),y(i),y(i+\productofprimes),...=n$.

\begin{lemma}\label{lem:interpretation2-inductive}
$\ConstraintIndex{i}{l}$ is inductive.
\end{lemma}
\begin{proof}
Let $u,v\in\configurations$ with $u\in\interpretation_2(\ConstraintIndex{i}{l})$ and $(u,v)\in\trafun$. We have to show that $v\in\interpretation_2(\ConstraintIndex{i}{l})$.

We start with $\trafun_{\mathit{mark}}$. If $\interpretation_2$ accepts $u$ in line~\ref{line:i-acc-1}, then it will accept $v$ in the same way, as $\trafun_{\mathit{mark}}$ only selects primes $p_j$ with $j>J(u)$ and only unmarks cells. If it accepts $u$ in line~\ref{line:i-acc-2}, then $\trafun_{\mathit{mark}}$ will select a remainder $r_j$ where $j:=J(u)+1$. If $r_j\not\equiv i\pmod{p_j}$, then $v$ is accepted in line~\ref{line:i-acc-1} (every $k$ with $\ConstraintIndex{i}{l}(k)=n$ has $k-s\equiv i\pmod{p_j}$). Otherwise, let $k>s$. We consider three cases:
\begin{itemize}
\item If $\ConstraintIndex{i}{l}(k)<j-1$, then $u(k)$ is unmarked (as $u$ was accepted in line~\ref{line:i-acc-2}), and $v(k)$ is as well (as $\trafun_{\mathit{mark}}$ only unmarks cells).
\item If $\ConstraintIndex{i}{l}(k)=j-1$, then $k-s\not\equiv i\pmod{p_j}$ by definition of $\ConstraintIndex{i}{l}$, so $\trafun_{\mathit{mark}}$ unmarks $v(k)$.
\item If $\ConstraintIndex{i}{l}(k)>j-1$, then $u(k)$ is marked and $k-s\equiv i\pmod{p_j}$, which implies $u(k)=v(k)$.
\end{itemize}
We find that $\ConstraintIndex{i}{l}(k)>j-1$ iff $v(k)$ is marked, so $v$ will also be accepted in line~\ref{line:i-acc-2}.

The remaining transitions $\trafun_{\mathit{write}}$ and $\trafun_{\mathit{init}}$ are straightforward, as they reset the prime part and mark every cell.
\end{proof}

\subparagraph*{Concluding the proof.}
All that remains is combining $\interpretation_1$ and $\interpretation_2$.

\begin{lemma}\label{lem:separatorifnotaccept}
Let $(u,v)\in C_I\times C_U$. If $\TM$ does not accept, there exists an inductive $\constraint\in\constraints$ with $u\in\interpretation(\constraint)$ and $v\notin\interpretation(\constraint)$.
\end{lemma}
\begin{proof}
Let $\constraint_1$ denote the constraint given by Lemma~\ref{lem:interpretation1-separating}. Based on the structure of $\constraint_1$, there are two cases. We remark that we rely on the specific construction in the proof of Lemma~\ref{lem:interpretation1-separating}.

\textbf{The first case.}
We have $\constraint_1=\square^{s+k_1}x\square^{k_2}y\square^{k_3}$ for some $k_1,k_2,k_3$ (note $k_2=\productofprimes-2$ and $y=\delta(x)$). Let $i:=k_1+4+k_2+1$ denote the position of $y$ within the TM part. We set $\constraint_2:=\ConstraintIndex{i}{|\constraint_1|}$ and $\constraint:=(\constraint_1,\constraint_2)$.

First, we show $u\in\interpretation(\constraint)$. Here, $u\in\interpretation_1(\constraint_1)$ follows from Lemma~\ref{lem:interpretation1-separating}, and $u\in\interpretation_2(\constraint_2)$ holds due to $u$ starting with $0^s$ and having all cells marked. Second, $v\notin\interpretation(\constraint)$ follows immediately from Lemma~\ref{lem:interpretation1-separating}.

To complete the first case, we must show that $\constraint$ is inductive, so let $(u',v')\in\trafun$ with $u'\in\interpretation(\constraint)$. By Lemma~\ref{lem:interpretation2-inductive} we have $v'\in\interpretation_2(\constraint_2)$, so it suffices to show $v'\in\interpretation_1(\constraint_1)$.

As $\trafun_{\mathit{mark}}$ only modifies the prime part and whether a cell is marked, $v'\in\interpretation_1(\constraint_1)$ follows from $u'\in\interpretation_1(\constraint_1)$ if $(u',v')\in\trafun_{\mathit{mark}}$.

For $(u',v')\in\trafun_{\mathit{write}}$, observe that only line~\ref{line:w-write} can cause $\constraint_1$ to be violated, and only if we write to $v(i)$ (recall that writing only changes a cell if it previously contained $\square$). For this to occur, $u'\in C_{\mathit{good}}$ must hold, i.e.\ the prime part of $u'$ has precisely one bit set for every prime. In particular, we have $J(u')=n$. If $u'(k)\ne\constraint_2(k)$ for some $k\in[s]$, constraint $\constraint_2$ implies that $u'(i)$ is not marked (note $\constraint_2(i)=n$). Otherwise, $\constraint_2$ implies that $u'(k)$ is marked iff $\constraint_2(k)=n$. So both $u'(i-\productofprimes)$ and $u'(i)$ are marked, and there are no marked cells between them. Hence, $\trafun_{\mathit{write}}'$ will read $u'(i-\productofprimes)$ in line~\ref{line:w-read}. Let $[{b \atop x'}] := u'(i-\productofprimes-1)...u'(i-\productofprimes+2)$. There are three possibilities: If $x'=x$, then $v'(i)=\delta(x)=y$. If $x'\ne x$ and $x'\notin\Confalph^4$, then $\square$ appears in $x'$ and $v'(i)=\delta(x)=\square$. Otherwise, $x'\ne x$ and $x'\in\Confalph^4$. In all cases, we get $v'\in\interpretation_1(\constraint_1)$.

Now we consider $(u',v')\in\trafun_{\mathit{init}}$. We again observe that $\constraint_2$ can only be violated if $\trafun_{\mathit{init}}$ writes to $u'(i)$. Similarly to before, we use $u'\in\Cgood$ to conclude that $u'(i)$ is only marked if $u'(i-\productofprimes)$ is marked, but we have $i-\productofprimes>s+1$ from the definition of $\constraint_1$.

\textbf{The second case.}
This case is both simpler than the first, and in large part analogous to it. We have $\constraint_1=\square^{s+k_1}y\square^{k_2}$ with $k_1\le \productofprimes+1$. Again, we write $i:=s+k_1+1$ for the position of $y$ and set $\constraint_2:=\ConstraintIndex{i}{|\constraint_1|}$ and $\constraint:=(\constraint_1,\constraint_2)$. Observe that $y=\btape$ if $i=s+\productofprimes+1$ and $y=\bk$ otherwise.

Similar to before, the only difficult part is showing that $\constraint$ is inductive w.r.t.\ $\trafun_{\mathit{write}}$ and $\trafun_{\mathit{init}}$; we define $u',v'$ as before. For both types of transitions we again find that only writing to $u'(i)$ can violate $\constraint_1$. Using $u'\in\Cgood$ we then have that $i$ is the first marked position after $s+1$ (using $i-\productofprimes\le s+\productofprimes-\productofprimes$). So $\trafun_{\mathit{write}}$ in fact does not write to $i$ (line~\ref{line:w-pick2} picks a larger index). For $\trafun_{\mathit{init}}$ we have that $u'(s+1)$ is marked iff $s+1=i-\productofprimes$ iff $y=\btape$. So $\trafun_{\mathit{init}}$ writes $y$ to $u'(i)$, and constraint $\constraint_2$ holds.
\end{proof}

This yields the desired theorem.

\begin{theorem}
\textsc{AbstractSafety} is \EXPSPACE-hard.
\end{theorem}
\begin{proof}
We have constructed an RTS $\system$ which has size polynomial in $n$. By Lemmata~\ref{lem:unsafeifaccepts} and~\ref{lem:separatorifnotaccept} an inductive constraint separating $C_I$ and $C_U$ exists if and only if $\TM$ does not accept.
\end{proof}

{
\renewcommand{\TM}{\mathcal{M}}
\renewcommand{\Tapealph}{\Gamma}
\renewcommand{\Confalph}{\Lambda}
\newcommand{\Sep}{\#}
\section{Complexity of the separability problem}\label{app:separability-hardness}

\thmseparability*
\begin{proof}

Membership in \PSPACE\ and in \NP\ respectively, are easy to see. In both cases we guess a constraint $\constraint$ that separates $c$ from $c'$, and then check that $c \in \interpretation(\constraint)$,  $c' \notin \interpretation(\constraint)$, and that $\constraint$ is inductive, which can be done by checking that $\constraint \notin \overline{\Inductive}$. Since the automaton for $\interpretation$ is part of the input, and a NFA recognising $\overline{\Inductive}$ can be constructed in polynomial time, given an input of size $n$ the checks can be carried out in $O(q(n))$ space and $O(|\constraint| \cdot p(n))$ time for some polynomials $p$ and $q$. This proves 
\textsc{Separability}$\in \NPSPACE = \PSPACE$. However, it does not prove \textsc{Separability}$\in \NP$ because the length of $\constraint$ may not be polynomial in $n$. However, if $\trafun$ and $\interpretation$ are length-preserving, then we necessarily have $|c|=|c'|=|\constraint|$. So the algorithm also runs in polynomial time.

\medskip \noindent\textbf{\PSPACE-hardness.} For the \PSPACE-hardness proof, we reduce  from the following problem: Given a deterministic Turing machine $\TM$ with $n$ states, does it accept when run on an empty tape with $n-2$ cells?

Let $Q$ denote the states of $\TM$, $q_0,q_f\in Q$ the initial and final state, respectively, and let $\Tapealph:=\{B,0,1\}$ denote its tape alphabet, with $B$ being the blank symbol. We represent a configuration of $\TM$ by a word $\Sep\,\beta\,q\,\eta$ (of length $n$), where $\Sep$ is a separator, $\beta,\eta$ encode the tape contents to the left and right of the head, respectively, and $q$ is the current state of $\TM$. Let $\Confalph:=Q\cup\Tapealph\cup\{\Sep\}$ be the symbols used for this. We encode the unique (infinite) run of $\TM$ by concatenating the encodings of each configuration, e.g.\ $\alpha:=\alpha_0\alpha_1...$, where $\alpha_i$ represents the $i$-th configuration of $\TM$.

Observe that $\alpha_0=\Sep\,q_0\,B^{n-2}$, and that the symbol $\alpha(i)$ is determined by the four symbols $\alpha(i-n-1)...\alpha(i-n+2)$ and the transition relation of $\TM$. We write $\delta(x_1...x_4)$ to denote this mapping (i.e.\ $\alpha(i)=\delta(\alpha(i-n-1)...\alpha(i-n+2))$). Finally, note that $\TM$ accepts iff $q_f$ occurs in $\alpha$. 

We will define a regular transition system $\system = (\configurations, \trafun)$ with configurations $c_I,c_U\in\configurations$ of initial and final configurations, as well as a regular abstraction framework $(\configurations, \constraints, \interpretation)$. We then need to show that there exists a constraint $\constraint\in\constraints$ separating $c_I$ and $c_U$, iff $\TM$ accepts.

Let us now give the formal description of the $\system$ and $(\configurations, \constraints, \interpretation)$. We set $\configurations:=\Confalph^*$. Let $c\in\configurations$. A transition of the RTS nondeterministically picks a position $i$, reads the symbols $c(i-1)...c(i+2)$ and then replaces $c(i+n)$ with $\delta(c(i-1)...c(i+2))$. 

For the abstraction framework, we set $\constraints:=\Confalph^*$. A constraint $\constraint\in\constraints$ is interpreted as follows: If $\constraint$ does not contain $q_f$, $\interpretation(\constraint):=\emptyset$. Otherwise, $\interpretation(\constraint)$ contains precisely the prefixes of $\constraint$. Clearly, this can be done by a transducer with a constant number of states.

Let $c_I:=\alpha_0$ and $c_U:=q_f$. We remark that there are many possible choices for $c_U$, indeed any word which is not a prefix of $\alpha$ works. Also note that the RTS cannot execute any transition starting in $c_I$ — however, the existence of those transitions still restricts which constraints are inductive. We make the following claim:

\medskip\noindent\textit{Claim:}  Let $\constraint\in\constraints$ denote an inductive constraint with $c_I\in\interpretation(\constraint)$. Then $\constraint$ is a prefix of $\alpha$ containing $q_f$.

\medskip\noindent\textit{Proof of the claim:}  If $q_f$ does not appear in $\constraint$, then $\interpretation(\constraint)=\emptyset$, contradicting $c_I\in\interpretation(\constraint)$. It follows that $q_f$ occurs in $\constraint$, and thus $c_I\in\interpretation(\constraint)$ implies that $c_I=\alpha_0=\Sep\,q_0\,B^n$ is a prefix of $\constraint$.

It remains to show that $\constraint$ is a prefix of $\alpha$, so assume the contrary. Then there is a position $i$ with $\constraint(i)\ne\alpha(i)$. We fix a minimal such $i$. From the previous paragraph we get $i>n$, and thus $\constraint(i)\ne\alpha(i)=\delta(\alpha(i-n-1)...\alpha(i-n+2))=\delta(\constraint(i-n-1)...\constraint(i-n+2))$. This proves the claim. 

We can now view $\constraint$ as a configuration of $\system$, with $\constraint\in\interpretation(\constraint)$. By executing the transition that writes to position $i$, we move to a configuration $c\ne\constraint$, using the previous inequality. But then $c\notin\interpretation(\constraint)$, so $\constraint$ would not be inductive. This proves the claim. 

We can now prove that the reduction is correct.  If an inductive constraint separating $c_I$ and $c_U$ exists, then $\TM$ accepts. Conversely, if $\TM$ accepts, let $\constraint$ denote any prefix of $\alpha$ that contains $q_f$. In particular, $\alpha_0$ is a prefix of $\constraint$, and so $c_I \in \interpretation(\constraint)$. Further, since by definition $\interpretation(\constraint)$ contains precisely the prefixes of $\constraint$, and $q_f$ is not a prefix of $\constraint$, we have 
$c_U \notin \interpretation(\constraint)$. So $\constraint$ separates $c_U$ from $c_I$. 

\medskip\noindent\textbf{\NP-hardness.}  For \NP-hardness, we reduce from the $3$-colourability problem for graphs. Let $G=(V, E)$ be a graph with $V = \{1, \ldots, n\}$. We construct in polynomial time an RTS, a abstraction framework, and two configurations $c, c'$ such that $c'$ is separable from $c$ if{}f $G$ is 3-colourable.

\smallskip\noindent\textit{The RTS.} We encode an edge $\set{i, j} \in E$ as a word in $\set{0, 1}^{n}$ by setting the $i$-th and $j$-th letter to $1$ while all others are $0$. Let $C_E \subseteq \set{0, 1}^{n}$ be the set of encodings of all edges. We define $\configurations = C_E \cup \{0^n\}$. We define $\trafun = C_E \times C_E$. Intuitively, the transition system allows us to ``move'' from an edge of $E$ to any other edge. Clearly, it is an RTS recognised by a transducer with four states.

\smallskip\noindent\textit{The abstraction framework.} Let $\constraints = \set{R, G, B}^n$ (where $R,G,B$ stand for red,  gree, and blue) and define the interpretation $\interpretation$ as follows $(x_{1} \; \ldots \; x_{n}, A_1 \ldots A_n) \in \interpretation$ if and only if  there are exactly two indices $i, j$ such that $x_{i} = x_{j} = 1$ and $\constraint_{i} \neq \constraint_{j}$. In other words, a constraint is satisfied by all edges whose endpoints have different colours under $A$. It follows that a constraint $\constraint$ is inductive if{}f it defines a colouring of $G$. 

\smallskip\noindent\textit{The configurations $c$ and $c'$.} Let $c$ be the encoding on an arbitrary edge of $E$, and let $c'=0^n$. 

\smallskip\noindent\textit{Correctness.} If $G$ is colourable, then there exists at least one inductive constraint $\constraint$. Further, $c$ satisfies $\constraint$, and $c'$ does not, because $c'$ does not satisfy any constraints at all. So $c'$ is separable from $c$. 
If $G$ is not colourable, then no constraint of $\constraints$ is inductive, and so $c'$ is not separable from $c$. 
\end{proof}

\section{Reducing length-preserving separability to SAT} \label{app:reductiontoSAT}
A explained in the main text of the paper, it suffices to construct a Boolean formula whose satisfying assignments are the words rejected by a given NFA.  Let $(Q, \Sigma, \delta, Q_{0}, F)$ be an NFA.  We introduce some propositional variables:

\medskip

\begin{itemize}
\item Variables $(\sigma, i) \in \Sigma \times \set{1, \ldots, n}$ encode the word by setting $(\sigma, i)$ to true if and only if the $i$-th letter of the word is $\sigma$. For every $i$, exactly one variable in the set $\set{(\sigma, i): \sigma \in \Sigma}$ can be made true by any satisfying assignment for the formula. It is straightforward to encode this as a formula $\psi_i$. 
\item  Variables $(q, i) \in Q \times \set{0, \ldots, n}$ encode that  state $q$ can be reached in the NFA after having read the first $i$ letters of the word.
\end{itemize}
Consider the formula
\begin{equation}
    \varphi = \bigwedge_{i=1}^n \psi_i \wedge \bigwedge \limits_{q\in Q_{0}} (q, 0) \land \bigwedge\limits_{q\in Q\setminus Q_{0}} \lnot (q, 0)
    \land \bigwedge\limits_{i = 1}^{n}\bigwedge\limits_{p \in Q} \left((p, i) \leftrightarrow \bigvee\limits_{p \in \delta(q, \sigma)} (q, i-1) \land (\sigma, i)\right).
\end{equation}
For any satisfying assignment $\chi$ of this formula, we have $\chi((q, i)) = 1$ if and only if one can reach the state $q$ in the NFA after reading a word $\sigma_{1} \;\ldots \; \sigma_{i-1}$ where $\chi((\sigma_{j}, j)) = 1$ for all $1 \leq j < i$.
To enforce that the NFA must not accept the word, one adds the clause $\bigwedge\limits_{f\in F} \lnot (f, n)$.


\end{document}